\documentclass[journal,twoside,web]{ieeecolor}
\usepackage{xcolor}
\usepackage{generic}
\usepackage{cite}
\usepackage{amsmath,amssymb,amsfonts}
\usepackage{graphicx}
\usepackage{hyperref}
\usepackage{booktabs}
\usepackage{tikz}
\usetikzlibrary{shapes.misc, positioning, arrows.meta, calc}

\usepackage{xparse}
\usepackage{float}
\usepackage{amsfonts}
\usepackage{amssymb}
\usepackage{bbm}
\usepackage[Export]{adjustbox} %
\usepackage{nicematrix}
\usepackage{bm}

\usepackage{mathtools}
\usepackage{arydshln}

\usepackage{accents}
\newcommand{\ubar}[1]{\underaccent{\bar}{#1}}

\newcommand\numberthis{\addtocounter{equation}{1}\tag{\theequation}}

\NewDocumentCommand{\rF}{
    O{R} d()
}{
    \IfNoValueTF{#2}{
        \mathbb{#1}
    }{
        \mathbb{#1}_{#2}
    }
} 
\NewDocumentCommand{\cF}{O{C}}{\mathbb{#1}} %

\NewDocumentCommand{\ii}{s}{
    \IfBooleanTF{#1}{
        \mathbf{i}
    }{
        i
    }
}

\NewDocumentCommand{\tRe}{r()}{
    \mathrm{Re}\left(#1\right)
}
\NewDocumentCommand{\tIm}{r()}{
    \mathrm{Im}\left(#1\right)
}

\NewDocumentCommand{\sbound}{sd()}{
    \IfBooleanTF{#1}{

    }{
        \partial\,#2
    }    
}
\NewDocumentCommand{\sint}{sd()}{
    \IfBooleanTF{#1}{
        
    }{
        \mathrm{int}\,#2
    }
}

\NewDocumentCommand{\rSp}{sd()}{
    \IfBooleanTF{#1}{
        \IfNoValueTF{#2}{
            \mathrm{range space}
        }{
            \mathrm{Ran}\left(#2\right)
        }
    }{
        \IfNoValueTF{#2}{
            \mathrm{ran}
        }{
            \mathrm{ran}\,#2
        }
    }
}
\NewDocumentCommand{\nSp}{sd()}{
    \IfBooleanTF{#1}{
        \IfNoValueTF{#2}{
            \mathrm{null space}
        }{
            \mathrm{null}\left(#2\right)
        }
    }{
        \IfNoValueTF{#2}{
            \mathrm{ker}
        }{
            \mathrm{ker}\,#2
        }
    }
}
\NewDocumentCommand{\gasq}{O{\mathbf{\Sigma^2}}d()}{
    #1\left(#2\right)
}
\NewDocumentCommand{\nr}{O{\mathbf{W}}d()}{
    #1(#2)
}

\NewDocumentCommand{\nnr}{O{\mathbf{W}_{\mathrm{N}}}d()}{
    #1\left(#2\right)
}

\NewDocumentCommand{\maxph}{sO{\phi}d()}{
    \IfNoValueTF{#3}{
        \IfBooleanTF{#1}{
            #2_{\max}
        }{
            \bar{#2}
        }
    }{
        \IfBooleanTF{#1}{
            #2_{\max}\left(#3\right)
        }{
            \overline{#2}\left(#3\right)
        }
    }
}
\NewDocumentCommand{\minph}{sO{\phi}d()}{
    \IfNoValueTF{#3}{
        \IfBooleanTF{#1}{
            #2_\{\min\}
        }{
            \ubar{#2}
        }
    }{
        \IfBooleanTF{#1}{
            #2_\{\min\}\left(#3\right)
        }{
            \underline{#2}\left(#3\right)
        }
    }
}

\NewDocumentCommand{\Tr}{d()}{
    \mathsf{Tr}\,(#1)
}

\NewDocumentCommand{\hermp}{s O{H} r()}{
    \IfBooleanTF{#1}{
        (#3)_{\mathrm{#2}}
    }{
        #3_{\mathrm{#2}}
    }
}

\NewDocumentCommand{\shermp}{s O{S} r()}{
    \IfBooleanTF{#1}{
        \left(#3\right)_{\mathrm{#2}}
    }{
        #3_{\mathrm{#2}}
    }
}

\NewDocumentCommand{\mati}{s o}{
    \IfNoValueTF{#2}{
        \IfBooleanTF{#1}{
            \mathbb{I}
        }{
            I
        }
    }{
        \IfBooleanTF{#1}{
            \mathbb{I}_{#2}
        }{
            I_{#2}
        }
    }
} %
\NewDocumentCommand{\mato}{s}{
    \IfBooleanTF{#1}{
        \mathbf{0}
    }{
        0
    }
} %

\NewDocumentCommand{\inv}{d()}{
    \IfNoValueTF{#1}{
        ^{-1}
    }{
        \left(#1\right)^{-1}
    }
}

\NewDocumentCommand{\tp}{s d()}{
    \IfBooleanTF{#1}{
        \IfNoValueTF{#2}{
            ^{\mathsf{T}}
        }{
            \left(#2\right)^{\mathsf{T}}
        }
    }{
        \IfNoValueTF{#2}{
            ^\mathsf{T}
        }{
            \left(#2\right)^{\mathsf{T}}
        }
    }
}

\NewDocumentCommand{\ct}{s d()}{
    \IfBooleanTF{#1}{
        \IfNoValueTF{#2}{
            ^{\mathsf{H}} 
        }{
            \left(#2\right)^{\mathsf{H}}
        }
    }{
        \IfNoValueTF{#2}{
            ^{\mathsf{H}} 
        }{
            \left(#2\right)^{\mathsf{H}}
        }
    }
}

\NewDocumentCommand{\cj}{r()}{
    \overline{#1}
}

\NewDocumentCommand{\proj}{s O{}}{
    \IfBooleanTF{#1}{
        \IfNoValueTF{#2}{
            \bm{\Pi}
        }{
            \bm{\Pi}_{#2}
        }
    }{
        \IfNoValueTF{#2}{
            \Pi
        }{
            \Pi_{\mathrm{#2}}
        }
    }
}

\newcommand{\projdn}{\proj[\mathrm{B}\cF]}
\newcommand{\projup}{\tilde{\Pi}_{\mathrm{B}\cF}}
\newcommand{\projga}{\Pi_{v}}

\NewDocumentCommand{\iprod}{r() r() o}{
    \IfNoValueTF{#3}{
        \left \langle #1, #2 \right \rangle
    }{
        \left \langle #1, #2 \right \rangle_{#3}
    } 
}

\NewDocumentCommand{\gph}{sd()}{\IfBooleanTF{#1}{\mathcal{G}_{\mathsf{l}}(#2)}{\mathcal{G}(#2)}}
\NewDocumentCommand{\invgph}{sd()}{\IfBooleanTF{#1}{\mathcal{G}^{-1}_{\mathsf{l}}(#2)}{\mathcal{G}^{-1}(#2)}}

\usepackage{xparse}
\usepackage{accents}
\usepackage{mathrsfs}
\usepackage{multirow}

\newcommand{\orhp}{\cF_+}
\newcommand{\crhp}{\overline{\cF}_+}
\newcommand{\poreals}{\rF_+}
\newcommand{\nnreals}{\overline{\rF}_+}
\newcommand{\vare}{\varepsilon}
\newcommand{\svmax}{\overline{\sigma}}
\newcommand{\svmin}{\underline{\sigma}}
\newcommand{\phmax}{\overline{\phi}}
\newcommand{\phmin}{\underline{\phi}}
\newcommand{\fpb}{f_{\mathrm{PB}}}

\newcommand{\rhinf}{\mathcal{RH}_\infty}
\newcommand{\rration}{\mathcal{R}}

\newcommand{\grasm}{{\bm{\mathcal{G}_r}}}
\NewDocumentCommand{\sgph}{D[]{G}}{\bm{\mathcal{#1}}}
\NewDocumentCommand{\sinvgph}{D[]{G}}{\bm{\mathcal{#1}}^{-1}}

\newcommand{\wno}{\mathrm{wno}\,}
\NewDocumentCommand{\poleno}{D[]{}r()}{|\mathscr{P}_{#1}(#2)|}
\NewDocumentCommand{\poles}{D[]{}}{\mathscr{P}_{#1}}

\NewDocumentCommand{\subsp}{o}{\mathcal{\IfNoValueTF{#1}{S}{#1}}}

\NewDocumentCommand{\angmin}{o}{
    \underline{\alpha}\IfNoValueF{#1}{_{#1}}
}
\NewDocumentCommand{\angmax}{o}{
    \overline{\alpha}\IfNoValueF{#1}{_{#1}}
}

\NewDocumentCommand{\nyContour}{sD(){\epsilon}}{
    \IfBooleanTF{#1}{
        \mathcal{D}_{#2}
    }{
        \mathcal{D}
    }
}

\NewDocumentCommand{\convhull}{s r()}{
    \IfBooleanTF{#1}{
        \mathbf{co}\,\left\{#2\right\}
    }{
        \mathbf{co}\,#2
    }
}

\NewDocumentCommand{\parahull}{s r()}{
    \IfBooleanTF{#1}{
        \mathbf{el}\,\left\{#2\right\}
    }{
        \mathbf{el}\,#2
    }
}

\NewDocumentCommand{\conihull}{s r()}{
    \IfBooleanTF{#1}{
        \mathbf{cone}\,\left\{#2\right\}
    }{
        \mathbf{cone}\,#2
    }
}

\NewDocumentCommand{\epi}{s r()}{
    \IfBooleanTF{#1}{
        \mathbf{epi}\,\left\{#2\right\}
    }{
        \mathbf{epi}\,#2
    }
}

\NewDocumentCommand{\eigs}{s r()}{
    \IfBooleanTF{#1}{
        \sigma \left(#2\right)
    }{
        \Lambda (#2)
    }
}

\NewDocumentCommand{\scomp}{s d()}{
    \IfBooleanTF{#1}{
        \IfNoValueTF{#2}{
            ^{\mathsf{c}} 
        }{
            \left(#2\right)^{\mathsf{c}}
        }
    }{
        \IfNoValueTF{#2}{
            ^{\mathsf{c}} 
        }{
            \left(#2\right)^{\mathsf{c}}
        }
    }
}

\NewDocumentCommand{\numran}{D[]{\mathbf{W}}r()}{
    #1(#2)
}

\NewDocumentCommand{\vnumran}{D[]{}r()}{
    \mathbf{V}_{#1}(#2)
}

\NewDocumentCommand{\dwshell}{sD[]{\mathbf{DW}}r()}{
    \IfBooleanTF{#1}{\widetilde{#2}}{#2}(#3)
}

\NewDocumentCommand{\srg}{sd[]r()}{
    \mathbf{SRG}\IfNoValueF{#2}{_{\IfBooleanTF{#1}{\mathrm{#2}}{#2}}}(#3)
}

\NewDocumentCommand{\ssg}{d[] r()}{
    \mathbf{SSG}\IfNoValueF{#1}{_{\mathrm{#1}}}(#2)
}

\NewDocumentCommand{\region}{sD[]{R}d()}{
    \IfBooleanTF{#1}{
        \mathbf{#2}^{\mathsf{c}}
    }{
        \mathbf{#2}
    }
    \IfNoValueTF{#3}{}{
        (#3)
    }
}

\NewDocumentCommand{\parab}{sD[]{}d()}{
    \IfBooleanTF{#1}{
        \mathbf{P}^{\mathsf{c}}
    }{
        \mathbf{P}
    }
    \IfNoValueTF{#2}{}{_{#2}}
    \IfNoValueTF{#3}{}{
        [#3]
    }
}

\NewDocumentCommand{\ball}{sD[]{}d()}{
    \IfBooleanTF{#1}{
        \mathbf{B}^{\mathsf{c}}
    }{
        \mathbf{B}
    }
    \IfNoValueTF{#2}{}{_{#2}}
    \IfNoValueTF{#3}{}{
        [#3]
    }
}

\NewDocumentCommand{\plane}{sD[]{}d()}{
    \IfBooleanTF{#1}{
        \mathbf{H}^{\mathsf{c}}
    }{
        \mathbf{H}
    }
    \IfNoValueTF{#2}{}{_{#2}}
    \IfNoValueTF{#3}{}{
        (#3)
    }
}

\usepackage{adjustbox}
\usepackage{subcaption}
\usepackage{makecell}
\usepackage[short]{optidef}

\usepackage[norelsize,ruled,lined,linesnumbered]{algorithm2e}
\IncMargin{0em}
\usepackage{etoolbox}
\makeatletter
\patchcmd{\@algocf@start}%
    {-1.5em}%
    {0pt}%
    {}{}%
\makeatother

\newtheorem{theorem}{Theorem}
\newtheorem{lemma}{Lemma}
\newtheorem{prop}{Proposition}
\newtheorem{corol}{Corollary}
\newtheorem{fact}{Fact}
\newtheorem{example}{Example}
\numberwithin{theorem}{section}

\newtheorem{defn}{Definition}
\newtheorem{remark}{Remark}
\numberwithin{remark}{section}

\usepackage[capitalise]{cleveref}
\crefname{defn}{Def.}{Defs.}
\crefname{prop}{Prop.}{Props.}
\crefname{corol}{Cor.}{Cors.}
\crefname{theorem}{Thm.}{Thms.}
\crefname{lemma}{Lem.}{Lems.}
\crefname{fact}{Fact}{Facts}
\Crefname{fact}{Fact}{Facts}
\crefname{section}{Sec.}{Secs.}

\hypersetup{hidelinks=true}
\usepackage{textcomp}
\def\BibTeX{{\rm B\kern-.05em{\sc i\kern-.025em b}\kern-.08em
    T\kern-.1667em\lower.7ex\hbox{E}\kern-.125emX}}
\begin{document}
\title{The Phantom of Davis-Wielandt Shell: A Unified {Perspective} for Graphical Stability Analysis of MIMO LTI Systems}
\author{Ding Zhang,
Xiaokan Yang, 
{\color{black} Axel Ringh}, 
{\color{black} Li Qiu, \IEEEmembership{Fellow, IEEE}}
\thanks{This work was  partially supported by the Wallenberg AI, Autonomous Systems and Software Program (WASP) funded by the Knut and Alice Wallenberg Foundation and the Swedish Research Council (VR) under grant 2024-05776, as well as by the Hong Kong Research Grants Council under the projects GRF 16203223 and GRF 16206324. 
}
\thanks{Ding Zhang is {{with the School of Engineering, Australian National University, Canberra, Australia}} (e-mail: ding.zhang@connect.ust.hk).}
\thanks{Xiaokan Yang is with the Department of Mechanics and Engineering Science, Peking University, Beijing, China (e-mail: yxkan21@stu.pku.edu.cn).}
\thanks{Axel Ringh is with the Department of Mathematical Sciences, Chalmers University of Technology and University of Gothenburg, Gothenburg, Sweden (e-mail: axelri@chalmers.se).}
\thanks{Li Qiu is with the School of Science and Engineering, The Chinese University of Hong Kong, Shenzhen, Guangdong, China (e-mail: qiuli@cuhk.edu.cn).}
}

\maketitle
\begin{abstract}
This paper presents a unified {{perspective}} based on Davis-Wielandt (DW) shell {{of linear subspaces}} for graphical stability analysis of multi-input and multi-output linear time-invariant feedback systems. Connections between DW shells and various graphical representations, as well as gain and phase measures, are established through an intuitive geometric perspective. 
{{We establish a strict DW separation theorem which characterizes subspace separation and provides a basis for comparing the implications}} and relative conservatism of graphical separation conditions. {{In particular, we introduce the}} rotated scaled relative graph ($\theta$-SRG) {{and show that its separation is equivalent to DW shell separation and is therefore the}} least conservative {{among all the graphical conditions considered}}. 
{{We derive pointwise and boundary stability conditions for bi-component feedback interconnections and establish a weak DW separation result that relaxes the strict separation requirement at open-loop poles. We also discuss applications to robust stability analysis and provide a numerical example illustrating how DW shells serve as a more reliable robustness indicator.}}
\end{abstract}


\section{Introduction}
\label{sec:introduction}
\IEEEPARstart{G}{raphical} system representations and stability conditions are intuitive and insightful tools favored by control theorists and practitioners. Notably, there has been a recent surge in the exploration of {{such tools for}} multi-input and multi-output (MIMO) linear and nonlinear systems, aimed at enabling graphical stability analysis of {{these}} systems.

The Nyquist plot is a {{classical}} graphical tool in the frequency-domain analysis and design of single-input and single-output (SISO) systems \cite[Chs.~10-14]{astromFeedbackSystemsIntroduction2020}. The closed-loop stability (and instability) can be assessed through the Nyquist plot of the \emph{return ratio} (a.k.a. loop transfer function). {{Various}} stability margins can be read off from {{it}}. The {{useful}} Bode plot can also be regarded as {{its}} frequencywise portrait. Furthermore, when mapped onto the {{Riemann}} sphere, the chordal distance between the Nyquist plots of systems provides a meaningful distance {{metric, which}} can be used to quantify robustness and {{is }}often maximized in SISO $\mathcal{H}_\infty$ control problems \cite[Secs.~9.4{{-9.9}}]{vinnicombeUncertaintyFeedbackInfinity2001,qiuIntroductionFeedbackControl2009}.
{{Undoubtedly, }}the Nyquist plots offer rich geometric {{insights into closed-loop}} properties of {{SISO}} systems{{, particularly robustness}}.

For MIMO linear time-invariant (LTI) systems, 
the generalized Nyquist criterion {{characterizes feedback}} stability {{through the}} eigenloci of the \emph{return ratio}, i.e., its eigenvalue trajectories along the Nyquist contour \cite{macfarlaneGeneralizedNyquistStability1977,desoerGeneralizedNyquistStability1980}.
{{For precisely known loop components, stability can also be assessed by a direct calculation of the closed-loop poles from the state-space realizations \cite[Lem.~5.2]{zhouRobustOptimalControl1995}. However, the main interest to the community currently lies in loops with uncertain components or of massive scale, for which direct closed-loop pole calculation becomes impractical. Also,}}
 unlike in the SISO case where the Nyquist plot of the return ratio can be constructed from those of loop components {{through}} scalar arithmetic rules $|ab|=|a||b|, \angle ab = \angle a + \angle b$, there {{is no}} obvious link between the eigenloci of the return ratio and those of {{individual}} loop component{{s}} in MIMO systems.
{{Moreover, eigenloci alone can give misleading indications of robustness \cite[Ex.~2.1.1]{greenLinearRobustControl2012}.}}

{{These limitations motivate alternative system representations}}, which are less precise than eigenloci but more informative on the {{closed-loop stability}} from individual components{{, and thus better suited to the analysis of uncertain feedback systems and large-scale dynamical networks}}.
A {{classical}} example in this category is the singular value-based gain, a notion woven into the $\mathcal{H}_{\infty}$ robust control {{theory}}~\cite{zhouRobustOptimalControl1995}. 
    At each frequency, the product of the largest singular values of loop components bounds the magnitudes {{of the}} return ratio {{eigenvalues}}.
Parallel to this, considerable efforts have been devoted to developing phase counterparts. 
    These include the sectorial phases \cite{wangPhasesSemiSectorialMatrix2023,chenPhaseTheoryMultiinput2024}, derived from the numerical range (NR) for semisectorial matrices, 
    {{and more}} recent developments {{applicable to all matrices,}} {{such as}} singular angles \cite{chenSingularAngleNonlinear2025} and segmental phases \cite{chenCyclicSmallPhase2025}, both derived from the normalized numerical range (NNR), as well as phases {{carried by}} the scaled relative graph (SRG) \cite{baron-pradaMixedSmallGain2025} and its signed version (SSG) \cite{eijndenPhaseScaledGraphs2025}. 
    The loop sums of phase {{bounds}} provide conic bounds{{ (i.e., argument bounds)}} on the return ratio eigen{{\kern0pt values}}.

{{The aforementioned phase notions are extracted from distinct 2-D sets such as NR, NNR, SRG, and SSG, whereas the SRG and SSG also embed the singular-value-based gain.}} The{{se}} 2-D sets are of interest in their own right%
{{, as t}}hey often {{yield}} less conservative {{stability}} conditions than those based solely on gains or phases. 
In particular, the SRG, a concept introduced in the optimization community \cite{hannahScaledRelativeGraph2016,ryuScaledRelativeGraphs2022} and later brought into control theory by \cite{chaffeyGraphicalNonlinearSystem2023,patesScaledRelativeGraph2021}, has attracted considerable attention. It has since inspired a line of follow-up works, including the extension in \cite{eijndenPhaseScaledGraphs2025} to incorporate missing phase lead/lag information into SRG, the customization for linear systems in \cite{chenGraphicalDominanceAnalysis2025, baron-pradaStabilityResultsMIMO2025}, which uses the less restrictive frequencywise SRGs of transfer matrices rather than considering the SRG of linear systems as operators on {{time-domain}} signal spaces, {{and enhancement through inner product scaling with structured weights\cite{grootExploitingStructureMIMO2025}}}.

That said, despite the proliferation of the{{se}} graphical representations and {{the corresponding stability conditions derived therefrom}}, {{their}} connection and distinction remain elusive. In this work, we bridge this gap by leveraging the concept of Davis-Wielandt (DW) shell to {{establish}} a unified {{perspective}} that derives {{these graphical}} sets and conditions via intuitive geometric transformations. This {{sheds light on}} the equivalences and implications among various {{graphical}} conditions. Furthermore, it leads us to a $\theta$-variant of SRG, yield{{ing one of}} the least conservative {{graphical}} condition{{s}}.
{{T}}he main technical contributions of this paper are summarized as follows:
\begin{itemize}
    \item We {{re-emphasize the Riemann-type formulation}} of the DW shell {{for a class of Grassmannian \emph{subspaces}, establish its properties, and generalize}} the {{DW separation}} condition in \cite{lestasNetworkStabilityGraph2011,lestasLargeScaleHeterogeneous2012}.
    \item We {{reveal}} the {{geometric connections}} between 3-D DW shells {{(as a unifying host) and existing lower-dimensional graphical system representations}}, and compare their induced {{separation}} conditions. From this {{unifying}} viewpoint, we examine {{their}} respective {{properties, subtleties concerning the $0$ (system zeros) and $\infty$ (system poles),}} and {{relative}} conservatism.
    \item We {{also introduce}} the $\theta$-SRG, {{establish its properties, and show that its separation with suitable $\theta$ is equivalent to DW-shell separation}}. 
    {{For bi-component feedback systems, we derive graphical stability certificates and also establish a weak separation result that renders the strict separation results less restrictive}}.
    \item {{We discuss the applicability of {{this set of graphical stability results}} to robust stability analysis and demonstrate {{its potential}} with a numerical example.}}
\end{itemize}  

The rest of this paper is organized as follows: 
{{\cref{sec:prob_form} presents a landscape of graphical stability analysis of MIMO LTI systems, emphasizing robust stability and its connection with pointwise graph separation.}}
\cref{sec:dw-shell-prelim} {{develops}} basic aspects of DW shells of {{subspaces}}, including a deferred literature review. 
\cref{sec:dw-phantom} {{connects}} DW shells {{to}} lower-dimensional graphical representations{{, which includes the introduction of $\theta$-SRG and development of its properties through geometry of DW shells and projections. Then \cref{sec:sep-cond} establishes and compares}} their associated separation conditions. 
\cref{sec:stability-ana} develops graphical feedback stability conditions {{and discusses their usage in robust stability analysis, supported by a numerical example.}}
Finally, the paper concludes in \cref{sec:concl}.

\emph{Notation.}
Let $\ii* = \sqrt{-1}$ be the imaginary unit{{, and}} $\mathbb{R}$, $\mathbb{C}$
{{be}} the sets of real numbers and complex numbers, respectively.
The sets of positive and nonnegative reals are denoted by $\mathbb{R}_+$ and $\overline{\mathbb{R}}_+$. {{The open and closed right-half planes are denoted by $\orhp$ and $\crhp$.}}
{{For}} a matrix $A$, we denote its transpose, conjugate, {{and}} conjugate transpose
by $A\tp, \cj(A)$, $A\ct*$,
respectively. 
{{Its kernel and image spaces are denoted by $\nSp(A)$ and $\rSp(A)$.}} 
{{The graph and inverse graph of $A \in \cF^{m\times n}$ are subspaces of $\cF^{n+m}$:}}
$\gph(A):=\{{{
    [u\ct\hspace{.3em}y\ct]\ct}}: y = Au, u \in \cF^{n}\} = \rSp({[\mati\hspace{.5em}A\ct]\ct})$,
{{$\invgph(A) := \{
    [y\ct\hspace{.3em}u\ct]\ct: y = Au, u \in \cF^{n}\}
    = \rSp([A\ct\hspace{.3em}\mati]\ct)$}}.
{{We denote by $\oplus$ and $\otimes$ the direct sum and the Kronecker product}}.
{{For a square matrix $A$, w}}e denote by $\Lambda(A)$ the set of its eigenvalues. 
{{It}} admits a {{unique}} Toeplitz decomposition: $A = \hermp(A) + \ii*\shermp(A)$ where both $\hermp(A) = (A+A\ct)/2$ and $\shermp(A)=(A-A\ct)/(2\ii*)$ are Hermitian.
$\|\cdot\|$ refers to the Euclidean vector norm. {{For a finite set, $|\cdot|$ denotes its cardinality. 
For sets $S_1$ and $S_2$, $S_1\setminus S_2$ denotes the set difference $\{x \in S_1: x\notin S_2\}$. 
Throughout, we take the convention that 
    $\cj(\infty) = \infty, \frac{\cF}{\infty} = 0, \frac{(\cF\,\cup\,\{\infty\})\setminus\{0\}}{0} =\infty$,
    and $\frac{0}{0},\frac{\infty}{\infty}$ are not allowed}}.

{{We treat}} $\cF\times \rF$ as a real inner product space endowed with the inner product $\iprod(({{w_1,v_1}}))(({{w_2,v_2}})):=\tRe(\cj(w_1) w_2) + v_1 v_2$. 
The three axes of $\cF\times\rF$ are referred to as the $\mathrm{Re}$-, $\mathrm{Im}$-, and $v$-axes, with the $\mathrm{Re}$-$\mathrm{Im}$ plane taken as horizontal and the $v$-axis as vertical.
A hyperplane $\plane({{\eta}},x_0)$ is determined by its normal vector ${{\eta}}$ and a point $x_0$ on the plane: $\plane({{\eta}},x_0):=\{x:\iprod({{\eta}})(x-x_0) = 0\}$. 
    We assign special notations to two particular classes of hyperplanes: 
        the vertical ones $\plane[\theta,d]:= \plane((e^{\ii*\theta},0),(de^{\ii*\theta},0))$ where $\theta \in \rF, d \in \nnreals$, and
        the horizontal ones $\plane[\gamma] = \plane((0,-1),(0,\gamma))$ where $\gamma \in \rF$. 
$\parab[a]\,(a\geqslant0)$ denotes a paraboloid in $\cF\times \rF$ that passes through $(a,1)$: $\parab[a] = \{(z,\nu): |z|^2 = a^2 \nu\}$.
{{A ball centered at the origin with radius $r>0$ is denoted by $\ball[r]$.}}
For a set $S\subseteq \cF[H]\times \rF$ where $\cF[H]$ is a real vector space, $\epi(S)$ refers to the Minkowski sum $S+(\{0\}\times (\nnreals{{\cup\{\infty\}}}))$.
$\sbound(S)$ and $\sint(S)$ denote the boundary and interior of $S$, respectively. The set $S\scomp$ refers to the complement of $S$ with respect to its ambient space. The convex, conic, and {{elliptic}} hulls of $S$ are defined as: 
    {{$\convhull(S) := \{\sum_{i=1}^{m} \lambda_i u_i:  u_i \in S, \lambda_i \geqslant 0, \sum_{i=1}^m\lambda_i = 1\}$ where $m$ ranges over all positive integers,}}
    $\conihull(S) := \{\lambda v: v \in S, \lambda \in \poreals\}$, and
    $\parahull(S):=\{
            {{(\frac{2\lambda x}{\lambda^2(1+\nu)+(1-\nu)}, \frac{\lambda^2(1+\nu)-(1-\nu)}{\lambda^2(1+\nu)+(1-\nu)})}}:
                (x,\nu)\in S, \lambda \in \poreals \vphantom{\frac{2\lambda x}{\lambda^2(1+\nu)+(1-\nu)}}\}$.

{{\ignorespaces
\section{Graphical Approaches to Feedback Stability Analysis of MIMO LTI Systems} \label{sec:prob_form}
{{\subsection{Internal stability and equivalent characterizations}}} \label{sec:int-stab-char}
Let $\rration$ be the field of real rational functions. Functions in $\rration$ that are
 bounded on $\orhp$ (equivalently, those that are proper and stable) form the ring $\rhinf$.
A superscript $m\times n$ on these symbols denotes the set of
$m$-by-$n$ matrices over the corresponding field or ring. For notational simplicity, the generalized frequency $s$ in transfer matrices is omitted when no ambiguity arises.
}}
\begin{defn}[{Internal Stability}] \label{def:internal_stab}
    Consider a {{positive}} feedback {{interconnection}} of LTI components modelled by transfer matrices {{$G(s) \in \rration^{p\times m}, H(s) \in \rration^{m\times p}$}}. The interconnection, denoted by $G\#H$ and illustrated in \cref{fig:standard-feedback}, is said to be \emph{internally stable} if the map {{$R(G,H)$}} from the {{internal signals $(e_1,e_2)$ to the external inputs $(u_1,u_2)$}} has a proper and stable inverse, i.e., $\exists\, K \in \rhinf^{(p+m)\times (p+m)}$ such that
    \begin{align}
        K R(G,H) = I\text{, where } R(G,H):=\begin{bNiceArray}{c@{}c}
            I & -H(s) \\
            -G(s) & I
        \end{bNiceArray}.  \label{eq:fdack-inv}
    \end{align}
\end{defn}
{{\ignorespaces
\begin{figure}[htbp]
    \vspace{-.3cm}
\subfloat[A standard feedback configuration]{
    \label{fig:standard-feedback}
    \makebox[.52\columnwidth][c]{
        \adjincludegraphics[Clip ={0.05\width} {.08\height} {0.42\width} {0.08\height}, height=2.3cm]{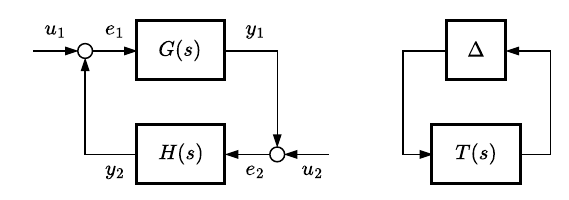}
    }
}
\subfloat[A robust stability setting]{
    \label{fig:robust-feedback}
    \makebox[.42\columnwidth][c]{
        \adjincludegraphics[Clip ={0.65\width} {.08\height} {0} {0.08\height}, height=2.3cm]{figure/diagram_cl_rs.pdf}
    }
}
\vspace{-.1cm}
\caption{Diagrams of positive feedback interconnections.}
\label{fig:feedback_system} \vspace{-.2cm}
\end{figure}
If $R(G,H)$ is nonsingular at some $s \in \cF$, then there exists a unique real rational $K$ satisfying (\ref{eq:fdack-inv}). This $K$ need not be proper or stable. In this case, we call the feedback system \emph{well-posed} and use $G\#H$ to denote both the interconnection and the unique $K$ (often known as the \emph{Gang of Four}).
There are several established ways to characterize the internal stability, and we recall two relevant ones here for completeness. 
\subsubsection{Algebraic characterization} 
For a real rational $P$ and a region $\mathbf{R}\subseteq \cF \cup\{\infty\} $, let $\poles[\mathbf{R}](P)$ denote the multiset of poles of $P$ that lie in $\mathbf{R}$, counted with the McMillan multiplicities. 
Note that $P$ always has a stable fractional representation $P = N_P D_P\inv$ where $[D_P;N_P]$ is a right coprime pair over $\rhinf$. 
    The non-minimum-phase zeros of $\det(D_P)$ counting multiplicities are precisely $\poles[\crhp](P)$, and the relative degree of $\det(D_P)$ equals $|\poles[\infty](P)|$. 
    Such a determinant of the denominator of a coprime factorization actually gives a characteristic determinant of $P$ with respect to $\rhinf$ \cite[Def.~4.3.2]{vidyasagarControlSystemsSynthesis2011a}.

For the positive feedback interconnection of $G$ and $H$, consider the stable right coprime factorizations of its loop components: $G = N_G D_G\inv$ and $H = N_H D_H\inv$. Then $R(G,H)$ admits a right coprime factorization over $\rhinf$: $D_{G\#H} N_{G\#H}\inv$, with 
\begin{align*}
    D_{G\#H} := \begin{bNiceArray}{c@{\hspace{.3em}}c}
        D_G & -N_H \\
        -N_G & D_H 
    \end{bNiceArray},
    \quad
    N_{G\#H} := 
    \begin{bNiceArray}{c@{\hspace{.3em}}c}
        D_G & \mato \\
        \mato & D_H
    \end{bNiceArray}. \numberthis \label{eq:coprime-def}
\end{align*}
The feedback interconnection is well-posed if and only if $\det(D_{G\#H})$ is not identically zero. 
When this condition holds, $G\#H$ admits a stable right coprime factorization given by $N_{G\#H}D_{G\#H}\inv$.
Consequently, $G\#H$ is internally stable if and only if $D_{G\#H}$ is invertible in $\rhinf^{(p+m)\times (p+m)}$, or equivalently, $\det(D_{G\#H})$ is invertible in $\rhinf$.

\subsubsection{Nyquist-type characterization} The open- and closed-loop characteristic determinants are related as follows \cite[p.~39]{greenLinearRobustControl2012}:
\begin{align}
    \det(I-GH) = \frac{\det(D_{G\#H})}{\det(D_G)\det(D_H)} \in \rration. \label{eq:poly-relation}
\end{align}
It follows that every zero of $\det(I-GH)$ in $\crhp\cup\{\infty\}$ is also a zero of $\det(D_{G\#H})$, with at least the same multiplicity. It is then necessary for internal stability that we require $\det(I-GH)$ does not have zeros on $\ii*\rF\cup\{\infty\}$. 

Though technically unnecessary, for practical considerations, we henceforth assume both loop components $G$ and $H$ are proper and thus $\det(D_G), \det(D_H)$ are both biproper.
As in the SISO Nyquist criterion, 
    we construct a clockwise-oriented contour $\nyContour$ consisting of a sufficiently long segment of the imaginary axis, indented into $\orhp$ by sufficiently small semicircles around the points in $\poles[\ii*\rF](G)\cup\poles[\ii*\rF](H)$, and closed by a sufficiently large semicircle centered at the origin in $\orhp$. 
Under the preceding nonvanishing requirement and by construction, $\det(I-GH)$ is analytic and nonzero on $\nyContour \cup \{\infty\}$, and thus $\poleno[\nyContour \cup \{\infty\}](K)=0$.
    Further applying the argument principle \cite[p.~289]{churchillComplexVariablesApplications2013} to (\ref{eq:poly-relation}) along $\nyContour$ yields:
\begin{align*}
    \wno\det(I-GH) = \poleno[\orhp](K) - \poleno[\orhp](G) - \poleno[\orhp](H),
\end{align*} where $\wno \det(I-GH)$ denotes the net number of clockwise encirclements of the origin by $\det(I-GH)$ evaluated along $\nyContour$. Then $\poleno[\orhp](K) = 0$ if and only if this trajectory makes a net total of $\poleno[\orhp](G)+\poleno[\orhp](H)$ \emph{counter}-clockwise encirclements of the origin. 
Finally, fully establishing internal stability requires checking that no closed-loop poles occur at any point in the finite set
$\poles[\ii*\rF](G)\cup\poles[\ii*\rF](H)$.

The following lemma summarizes these characterizations.
\begin{lemma}[\cite{vidyasagarControlSystemsSynthesis2011a,vinnicombeUncertaintyFeedbackInfinity2001}]\label{lem:istab_char}
    Let $G\in \rration^{p\times m}$ and $H \in \rration^{m\times p}$ be proper. The following statements are equivalent:
    \begin{enumerate}
        \item $G\#H$ is internally stable;
        \item $D_{G\#H}$ is nonsingular for all $s \in \crhp\cup\{\infty\}$;
        \item $\det(D_{G\#H})$ is invertible in $\rhinf$ (i.e., biproper and minimum-phase).
        \item $\det(I-GH)\neq 0$ for all $s \in \ii*\rF\cup\{\infty\} \setminus(\poles[\ii*\rF](G)\cup\poles[\ii*\rF](H))$; the net number of counter-clockwise encirclements of the origin by the trajectory of $\det(I-GH)$ along $\nyContour$ equals the total number of open right-half plane poles of $G$ and $H$, counted with their McMillan multiplicities; and $\det(D_{G\#H}(s)) \neq 0$ at every $s\in\poles[\ii*\rF](G)\cup\poles[\ii*\rF](H)$.
    \end{enumerate}
\end{lemma}
\begin{remark}
    For 4), the Nyquist plot of $\det(I-GH)$ is related to the eigenloci of $L = -GH$ along $\nyContour$, and its winding number around the origin can be counted through the winding number of the eigenloci around $-1$ \cite{desoerGeneralizedNyquistStability1980,macfarlaneGeneralizedNyquistStability1977}.
\end{remark}
\subsection{Key points in graphical (robust) stability analysis}
In this work, we focus on \emph{graphical} system representations suited to \emph{robust} stability analysis.
We emphasize two key requirements: 
    (1) they should admit an explicit construction from the system model that provides geometric insights into stability; and 
    (2) they should have the ability to characterize uncertainties and relate individual component representations to feedback stability in a transparent way.

Although \cref{lem:istab_char} does not directly meet these requirements, it underpins all graphical stability conditions developed and discussed herein. 
We therefore first clarify the underlying technical problem that these conditions address, which is the graph separation between loop components.
Condition 2) of \cref{lem:istab_char}, through a similarity transformation via $\mati \oplus -\mati$, is equivalent to: for each $s\in \crhp\cup\{\infty\}$,
\begin{align*}
    \rSp(\begin{bNiceMatrix}
        D_G \\ N_G
    \end{bNiceMatrix}) \cap 
    \rSp(\begin{bNiceMatrix}
        N_H \\ D_H
    \end{bNiceMatrix}) = \{0\}. \numberthis \label{eq:pt-gsep}
\end{align*} 
This is in essence \emph{the pointwise separation of two subspaces} (void of nontrivial intersection) associated with $G$ and $H$. In the following sections, we introduce the DW shell of Grassmannian subspaces and use it to study such separation problems.
We then relate the DW shell to planar graphical representations and their associated system measures, and characterize the implications among the resulting conditions.

Apart from pointwise conditions on $\crhp\cup\{\infty\}$, we will investigate graphical conditions imposed only on the contour $\nyContour$. In \cref{sec:stability-ana}, we combine these conditions with a homotopy argument to obtain sufficient boundary conditions for stability of the feedback interconnection in \cref{fig:standard-feedback}. 
We also discuss how they enable robust stability analysis with a visual aid.

Finally, note that these boundary conditions are not inherently restricted to interconnections in \cref{fig:standard-feedback} with stable or marginally stable loop
components. Rather, they ensure zero winding numbers for specific transfer matrices.
As discussed in \cite[Sec.~1.3]{vinnicombeUncertaintyFeedbackInfinity2001}, such guarantees can support robust stability analysis of the uncertain feedback systems in \cref{fig:robust-feedback}, where the structure of the stable $T(s)$ depends on how the uncertainty $\Delta$ interacts with the nominal plant, and neither the plant nor $\Delta$ is required to be stable, but the number of unstable poles must be invariant when the plant is perturbed by $\Delta$.

}}

\section{The Davis-Wielandt Shell {of a Subspace}}
\label{sec:dw-shell-prelim}
The notion of DW shell (of a matrix, operator, or relation) traces back to the seminal {{work}} of Wielandt \cite{wielandtEigenvaluesSumsNormal1955} and Davis \cite{davisShellHilbertspaceOperator1968,davisShellHilbertspaceOperator1970}. {{Subsequent developments}} in the linear algebra community {{include}} \cite{liDavisWielandtShellsOperators2008,liEigenvaluesSumMatrices2008} {{and the shell is of persisting interest there}}. Meanwhile, on the control side, its special form {{for a frequencywise normal}} transfer matrix appears {{implicitly}} in \cite{jonssonScalableRobustStability2010}. {{An explicit formulation was first}} introduced, {{developed}} and applied to large-scale networks in{{\cite{lestasNetworkStabilityGraph2011,lestasLargeScaleHeterogeneous2012}}} and {{subsequently}} to communication networks in \cite{zhangLocalStabilityCongestion2025}. {{The DW shell has also been interpreted as combining}} the multivariable gain and sectorial phase \cite{zhaoWhenSmallGain2022}, and {{used}} to {{analyze}} robust stability {{under uncertainties described by gain and phase bounds}} \cite{srazhidinovComputationPhaseGain2023, liangFeedbackStabilityMixed2025}.

\subsection{Subspaces and two versions of the shell} \label{subsec:dw-def}
{{\ignorespaces
Let $\grasm(k,l)$ denote the Grassmannian of $k$-dimensional subspaces of $\cF^l$.
Each $\subsp[A] \in \grasm(n,2n)$ can be represented as $\subsp[A]=\rSp(X)$ where $X\in \cF^{2n\times n}$ has full rank. Partitioning $X$ into two square blocks $[X_1\ct\hspace{.5em}X_2\ct]\ct$,  we define the inverse subspace $\subsp[A]\inv := \rSp([X_2\ct\hspace{.5em}X_1\ct]\ct)$.
We will refer to $[X_1;X_2]$ as a \emph{basis pair} of $\subsp[A]$, thus $[X_2;X_1]$ is a basis pair of $\subsp[A]\inv$.
Note that if $X_1$ is invertible, $\subsp[A]$ is the graph of the matrix $A:=X_2 X_1\inv$, in which case we call $A$ a matrix (or Tucker\cite[p.~9]{rockafellarConvexAnalysis1970}) representation of $\subsp[A]$; and if $X_2$ is invertible, $\subsp[A]$ is the inverse graph of $X_1 X_2 \inv$.
Following Davis's work \cite{davisShellHilbertspaceOperator1968}, we introduce the DW shell of a subspace:
}}
\renewcommand{\arraystretch}{.7}
\begin{defn}[{DW Shell of a {{Subspace}}}] \label{defn:dw-shell}
    The DW shell of {{$\subsp[A] \in \grasm(n,2n)$}} is {{a 3-D set in $\cF\times\rF$, defined by}}
    \begin{align}
        \hspace{-.5em}
        \dwshell({{\subsp[A]}}):=& 
        \{{{\textstyle\left(\frac{2\iprod(u)(y)}{\|u\|^2+\|y\|^2},\frac{\|y\|^2-\|u\|^2}{\|y\|^2+\|u\|^2}\right)}}:
        {{\begin{bNiceMatrix}
            u \\ y
        \end{bNiceMatrix} \in \subsp[A]\setminus \{0\}}} \}.
        \label{eq:dw-defn}
    \end{align}
\renewcommand{\arraystretch}{1.0}
    {{\color{black}
        It can also be written as the following equivalent form in $\rF^3$:
        \begin{align*}
            &\left\{(\xi\ct E_1 \xi, 
            \xi\ct E_2 \xi,
            \xi\ct E_3 \xi
            ): \xi \in \subsp[A], \xi\ct\xi = 1\right\}, \numberthis \label{defn-dw-r3}\\
            &\qquad \text{where }
            E_1 = \begin{bNiceArray}{c@{\hspace{.5em}}c}
                \mato & \mati \\
                \mati & \mato
            \end{bNiceArray}, 
            E_2 = \begin{bNiceArray}{c@{\hspace{.5em}}c}
                \mato & -\ii*\mati \\
                \ii*\mati & \mato
            \end{bNiceArray},
            E_3 = \begin{bNiceArray}{c@{\hspace{.5em}}c}
                -\mati & \mato \\
                \mato & \mati
            \end{bNiceArray}.
        \end{align*}
    }}
\end{defn}
\begin{figure}[t]
    \vspace{-.3cm}
    \centering
    \subfloat[Paraboloidal version\label{fig:dw-parab}]{
        \makebox[.46\columnwidth][c]{
            \adjincludegraphics[Clip={0.01\width} {.13\height} {0.03\width} {0}, width=.35\columnwidth]{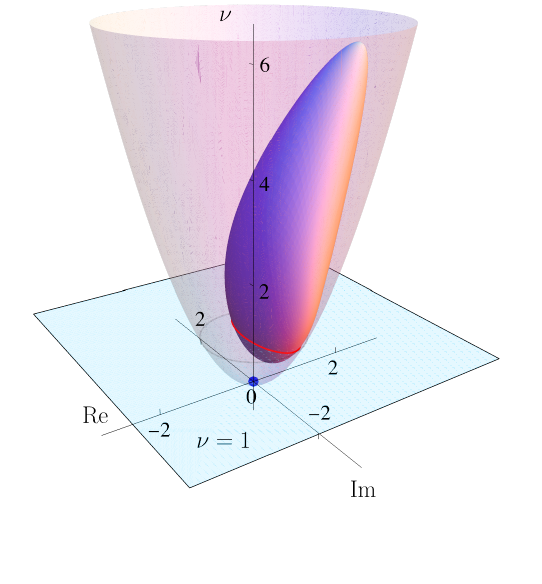}
        }
    }
    \subfloat[Unit ball (spherical) version\label{fig:dw-ball}]{
        \makebox[.52\columnwidth][c]{
            \adjincludegraphics[Clip ={0.04\width} {.25\height} {0.04\width} {0}, width=.4\columnwidth]{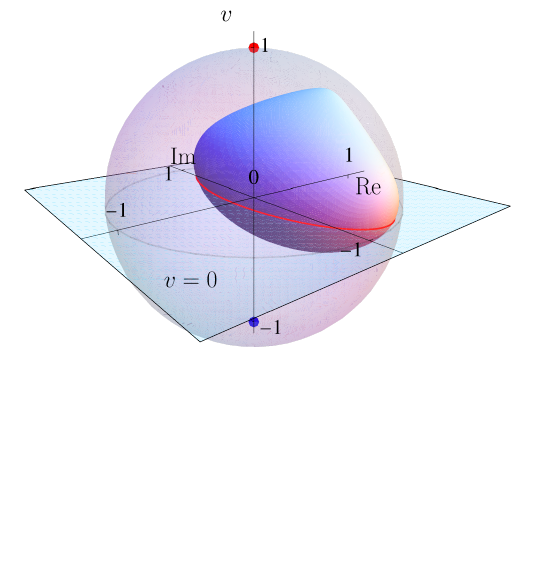}
        }
    }\vspace{-.1cm}
    \caption{Two versions of DW shells of the same subspace.} 
    \vspace{-1.0cm}
\end{figure}
{{\ignorespaces
\begin{remark}\label{rmk:dw-defn}
    We do not distinguish a subspace $\subsp[A]$ and its matrix representation $A$ (if exists). Thus, we take $\dwshell(A) = \dwshell(\subsp[A])$, and vice versa.
    By the Cauchy-Schwarz inequality, $\dwshell(\subsp[A])$ always resides in the closed unit ball (cf.~\cref{fig:dw-ball}). 
    The shell of a matrix used in \cite{liDavisWielandtShellsOperators2008,lestasLargeScaleHeterogeneous2012,zhangLocalStabilityCongestion2025} lies above a paraboloid (cf.~\cref{fig:dw-parab}). The unit-ball shell is the image of the paraboloidal shell under a linear fractional map $\fpb:(z,\nu)\mapsto (\frac{2z}{\nu+1},\frac{\nu-1}{\nu+1})$. Their resulting conditions are equivalent and we adopt the unit ball version for its symmetry in handling 0 and $\infty$, and adopt the subspace formulation to accommodate systems with poles at the frequency of interest.
\end{remark}

Note that $\fpb$ is a bijection between the extended paraboloid and the unit sphere $\sbound(\ball[1])$, with the origin (or $\infty$) mapped to the \emph{south pole $(0,-1)$} (or \emph{north pole $(0,1)$}). The map $\fpb$ also preserves lines (thus convexity of bodies contained in $\epi(\parab[1])$). In particular, any vertical ray, originated from $\nu=\infty$ and pointing downward, maps to a ray injecting from the north pole toward the complex plane. This gives rise to a stereographic projection from the unit ball to the extended complex plane:
\begin{align*}
    \projdn:\; \ball[1] \rightarrow \cF\cup\{\infty\},
    \, (w,v) \mapsto \frac{w}{1-v}.
\end{align*}
Its restriction to the unit sphere is a bijection, with inverse:
\begin{align*}
    \projdn\inv:\; \cF\cup\{\infty\} \rightarrow \sbound(\ball[1]),
    \, z \mapsto
    \left(\frac{2z}{|z|^2+1},
          \frac{|z|^2-1}{|z|^2+1}\right).
\end{align*} Throughout, $\projdn\inv$ denotes the inverse of the restriction of $\projdn$ to the unit sphere.
For brevity, we identify the $\cF\times\rF$ representation (\ref{eq:dw-defn}) with the $\rF^3$ representation (\ref{defn-dw-r3}) via $(w,v)\leftrightarrow[\tRe(w)\, \tIm (w)\,v]\tp$. Maps acting on either representation are understood through this identification.
}}

\subsection{Basic properties} \label{subsec:dw-basics}
The DW shell exhibits nice geometric properties while encompassing useful information of {{an input-output relation captured by a subspace. 
To see this, let us first define the spectrum $\eigs(\subsp[A])\subseteq \cF\cup\{\infty\}$ of a subspace $\subsp[A] \in \grasm(n,2n)$ as the set of generalized eigenvalues of its associated matrix pencil. Specifically,
let $[D;N]$ be a basis pair of $\subsp[A]$, 
\begin{align*}
    \eigs(\subsp[A]):= \left\{\frac{\beta}{\alpha}:
        \det(\alpha N - \beta D) = 0, \begin{bNiceMatrix}
            \beta \\ \alpha
        \end{bNiceMatrix} \in \cF^2 \setminus \{0\}
    \right\}. 
\end{align*} 
We call a subspace $\subsp[A]$ \emph{normal} if there exists a unitary $Q$ and nonsingular $Z$ such that $Q\ct N Z$ and $Q\ct D Z$ are both diagonal.
Again, whenever $\subsp[A]$ has a matrix representation, the spectrum and normality concepts coincide with the matrix version.
In addition, note that for subspaces considered here, $\eigs(\subsp[A])$ is either a point spectrum consisting of a finite set of points, or the entire extended complex plane if $\det(\alpha N - \beta D) \equiv 0$ for all $\alpha,\beta$.
The following lemma summarizes the properties that underpin our geometric intuitions.}}
\begin{lemma}[Basic Properties] \label{lem:dw-shell-basics}
    {{Let $\subsp[A] \in \grasm(n,2n)$.}}
    \begin{enumerate}
        \item $\dwshell({{\subsp[A]}})$ is always compact (i.e., closed and bounded).
        \item $\dwshell({{\subsp[A]}})$ is convex when $n=1$, or $n\geqslant3$. When $n=2$, $\dwshell({{\subsp[A]}})$ is an ellipsoid, possibly degenerate. 
        \item $\dwshell({{\subsp[A]}})$ is always contained in {{the unit ball $\ball[1]$}}. Moreover, $\dwshell({{\subsp[A]}})\cap \sbound(\ball[1]) = {{\projdn\inv\,}} \Lambda({{\subsp[A]}})$.
        \item $\dwshell({{\subsp[A]}}) = \convhull((\dwshell({{\subsp[A]}})\,\cap\, \sbound(\ball[1])))$ if ${{\subsp[A]}}$ is normal. {{The ``only if'' side holds if $n=2$ or $\eigs(\subsp[A])$ is a finite set}}.
    \end{enumerate}
\end{lemma}
\begin{remark} \label{rem:pf-basic-dw}
    {{The matrix version of }} 1)\,-\,3) can be found in \cite{liDavisWielandtShellsOperators2008}, and 4) is hidden in \cite[Cond.~(69)]{groneNormalMatrices1987}. {{Analogous arguments extend these results to subspaces.
    If $\eigs(\subsp[A])=\cF\cup\{\infty\}$, then $\dwshell(\subsp[A])$ is the entire unit sphere when $n=2$ and the entire unit ball when $n\geqslant3$.
    In both cases, Property 3) holds. For $n\geqslant3$, however, the convex-hull condition in 4) holds regardless of the normality of $\subsp[A]$. Thus, the finiteness assumption is essential for the converse.
    }}
\end{remark}

It is often handy to {{relate}} the DW shell of transformations on ${{\subsp[A]}}$ {{to}} $\dwshell({{\subsp[A]}})$ itself. {{We derive such relationships for several basic operations on subspaces.}}
{{\ignorespaces
    \begin{lemma} \label{lem:unary-op}
        Let $\subsp[A]\in \grasm(n,2n)$.
        \begin{enumerate}
            \item \emph{Simultaneous block unitary transformation}: For a unitary $U \in \cF^{n\times n}$, $\dwshell({(\mati_2 \otimes U)\subsp[A]}) = \dwshell(\subsp[A])$.
            \item \emph{Unitary block transformation}: For a unitary $U \in \cF^{2\times 2}$, there exists a unique real orthonormal matrix $Q(U)\in \rF^{3\times 3}$ such that $\dwshell((U \otimes \mati_n)\subsp[A]) = Q(U) \dwshell(\subsp[A])$.
                In particular, for $\theta\in\rF$,
                \begin{align*}
                    \hspace{-.55cm}
                    Q(\begin{bNiceArray}{c@{\hspace{.4em}}c}
                        1 & 0 \\
                        0 & e^{\ii*\theta}
                    \end{bNiceArray}) = 
                    \begin{bNiceArray}{c@{\hspace{.4em}}c@{\hspace{.4em}}c}
                        \cos \theta & -\sin \theta & 0 \\
                        \sin \theta & \cos \theta & 0\\
                        0 & 0 & 1
                    \end{bNiceArray}, 
                    Q(\begin{bNiceArray}{c@{\hspace{.4em}}c}
                        0 & 1 \\
                        1 & 0
                    \end{bNiceArray}) = 
                    1\oplus (-\mati_2).
                \end{align*}
            \item \emph{Scaling}: For $\gamma \in \poreals$, $\dwshell({(\mati \oplus \gamma \mati)\subsp[A]}) = f_\gamma \dwshell(\subsp[A])$,  
            \begin{align*}
             \text{with }  f_\gamma: (w,v)\mapsto {\textstyle \left(\frac{2\gamma w}{\gamma^2(1+v)+(1-v)}, \frac{\gamma^2(1+v)-(1-v)}{\gamma^2(1+v)+(1-v)}\right)}
            \end{align*} scaling a point along the elliptic arc connecting the north pole, south pole, and the point itself.
            \item \emph{Conjugation}: Let $\cj(\subsp[A]):=\{\cj(\xi): \xi \in \subsp[A]\}$, then $\dwshell(\cj(\subsp[A]))$ and $\dwshell(\subsp[A])$ are symmetric about the $\mathrm{Re}$-$v$ plane.
            \item \emph{Orthogonal complement}: The orthogonal complement $\subsp[A]^\perp$ of $\subsp[A]$ is in $\grasm(n,2n)$, and
                $\dwshell(\subsp[A]^\perp)= -\dwshell(\subsp[A])$.
        \end{enumerate}
    \end{lemma}
    \begin{proof}
        See Appendix~\ref{app:proof-dw-basic}.
    \end{proof}
    \begin{remark} \label{rem:dw-basic}
        Note that $\subsp[A]\inv = (\left[\begin{smallmatrix}
                        0 & 1 \\
                        1 & 0
                    \end{smallmatrix}\right]\otimes I)\subsp[A]$. By Property 2), $\dwshell(\subsp[A]\inv)$ can be obtained by rotating $\dwshell(\subsp[A])$ through $180^\circ$ about the $\mathrm{Re}$-axis (also obvious from the definition).
        When $\subsp[A]$ has a matrix representation $A$, most of these transformations can be easily related to their matrix counterparts. In particular,  
            $(\mati_2 \otimes U) \subsp[A]$ gives the graph of $U A U\ct$, 
            $\cj(\subsp[A])$ is the graph of $\cj(A)$,
            and $\subsp[A]^\perp$ is the inverse graph of $-A\ct$.
    \end{remark}

    \cref{lem:dw-shell-basics}-3) identifies the intersection of the DW shell and the unit sphere as the inverse stereographic projection of the spectrum. For system analysis, we are particularly interested in whether the DW shell contains the two poles. The following \cref{prop:n-s-pole} foreshadows relevant insights at the subspace level:
    \begin{prop} \label{prop:n-s-pole}
        For $\subsp[A] \in \grasm(n,2n)$ with basis pair $[D;N]$,
        the following statements are equivalent:
        \begin{enumerate}
            \item $\dwshell(\subsp[A])$ contains the north (respectively, south) pole;
            \item $D$ is singular (respectively, $N$ is singular);
            \item $\subsp[A]$ (respectively, $\subsp[A]\inv$) \emph{has no} matrix representation.
        \end{enumerate}
    \end{prop}
    \begin{proof}
        By (\ref{eq:dw-defn}), $\dwshell(\subsp[A])$ contains the north pole if and only if $\subsp[A]$ contains a vector $[u\ct\hspace{.5em}y\ct]\ct$ with $u = 0, y \neq 0$, or equivalently, the block $D$ in its basis pair is singular.
        Since the basis pair of a subspace is unique up to simultaneous right multiplication by a nonsingular matrix, these singularity conditions are independent of the chosen basis. 
        Thus, $D$ is singular if and only if $\subsp[A]$ admits no matrix representation.
        
        \cref{rem:dw-basic} informs $\dwshell(\subsp[A]\inv) = (1\oplus -1 \oplus -1) \dwshell(\subsp[A])$, which indicates the south pole is contained by $\dwshell(\subsp[A])$ if and only if the north pole is contained by $\dwshell(\subsp[A]\inv)$. 
        The south-pole statements then follow by applying the same argument to $\subsp[A]\inv$, for which $[N;D]$ gives a basis pair.
    \end{proof}
\begin{table*}[!htbp]
\caption{Definitions of 1-D and 2-D {{graphical}} representations of a square matrix $A$.} \label{tab:2d-set-defns}
\centering
\begin{NiceTabular}{@{\hspace{.5em}}l@{}|l@{}|p{11em}}
    \toprule
    1-D / 2-D representation & definition & related system concept \\ \midrule
    \begin{minipage}[t]{11em}
        squared gain \\ interval (1-D) 
    \end{minipage}
        & \begin{minipage}[t]{33em}
            $\displaystyle\gasq(A)=[\svmin(A)^2, \svmax(A)^2]$ with $\displaystyle \svmin(A) := \min \left\{\frac{\|y\|}{\|{{u}}\|}:
        \begin{bmatrix}
            {{u}} \\ y 
        \end{bmatrix} \in \gph(A) {{\setminus}} \{0\}\right\}$\\\hspace{2em} and $\displaystyle \svmax(A) := \max \left\{\frac{\|y\|}{\|{{u}}\|}: \begin{bmatrix}
                {{u}} \\ y
            \end{bmatrix} \in \gph(A)\setminus \{0\}\right\}$
        \end{minipage} 
        & $\mathcal{H}_\infty$ gain  \\ 
    \begin{minipage}[t]{11em}
        numerical range \\
        (2-D) \cite[Sec.~1]{hornTopicsMatrixAnalysis1994} 
    \end{minipage}
        & \begin{minipage}[t]{27em}
            $\displaystyle \nr(A):= \left\{\frac{\iprod({{u}})(y)}{\|{{u}}\|^2} : \begin{bmatrix}
                {{u}} \\ y
        \end{bmatrix} \in \gph(A), {{u}}\neq 0\right\}$
        \end{minipage} & sectorial phase \cite{chenPhaseTheoryMultiinput2024} \\
    \begin{minipage}[t]{11em}
        scaled relative graph \\ 
        (2-D) \cite{hannahScaledRelativeGraph2016,patesScaledRelativeGraph2021,chaffeyGraphicalNonlinearSystem2023} 
    \end{minipage}
        & \begin{minipage}[t]{31em}
            $\displaystyle\srg(A):= \left\{ \frac{\|y\|}{\|{{u}}\|}\exp\left(\pm \ii* \arccos\frac{\mathrm{Re}\iprod({{u}})(y)}{\|{{u}}\|\|y\|}\right):\begin{bmatrix}
                {{u}} \\ y
            \end{bmatrix} \in \gph(A), {{u}} \neq 0\right\}$
        \end{minipage} 
        & \begin{minipage}[t]{11em} 
            SRG phase \cite{baron-pradaMixedSmallGain2025}\\(a.k.a. singular angle \cite{chenSingularAngleNonlinear2025})
        \end{minipage}\\ 
    \begin{minipage}[t]{11em}
        signed scaled \\ graph (2-D)
    \end{minipage}
        & \begin{minipage}[t]{37em}
            $\displaystyle
                \ssg(A):= \left\{\frac{\|y\|}{\|{{u}}\|}\exp\left(\ii*\, \mathrm{sign}(\mathrm{Im}\iprod({{u}})(y)) \arccos\frac{\mathrm{Re}\iprod({{u}})(y)}{\|{{u}}\|\|y\|}\right):\begin{bmatrix}
                    {{u}} \\ y
                \end{bmatrix} \in \gph(A),{{u}} \neq 0\right\}$
        \end{minipage} 
        & SSG phase \cite{eijndenPhaseScaledGraphs2025} \\ & & \\
    \begin{minipage}[t]{11em}
        normalized numerical \\ range (2-D) 
    \end{minipage}
        & 
        \begin{minipage}[t]{30em} 
            $\displaystyle \nnr(A):=\left\{{{\frac{\iprod({{u}})(y)}{\|{{u}}\|\|y\|} }}: \begin{bmatrix}
                {{u}} \\ y
            \end{bmatrix} \in \gph(A), y\neq 0 \right\}$
        \end{minipage} 
        & segmental phase \cite{chenCyclicSmallPhase2025}
    \\\bottomrule
\Block[l]{1-3}{\textbullet$\hspace{.5em}\mathrm{sign}(x)$ takes value 1 if $x>0$, $-1$ if $x<0$. When $x=0$, we let $\mathrm{sign}(x)=\pm 1$.}
\end{NiceTabular}\vspace{-.4cm}\\
\end{table*}
\begin{figure*}[!t]
    \centering
    \adjincludegraphics[Clip={0} {0} {0.002\width} {0.012\height}, width=\textwidth]{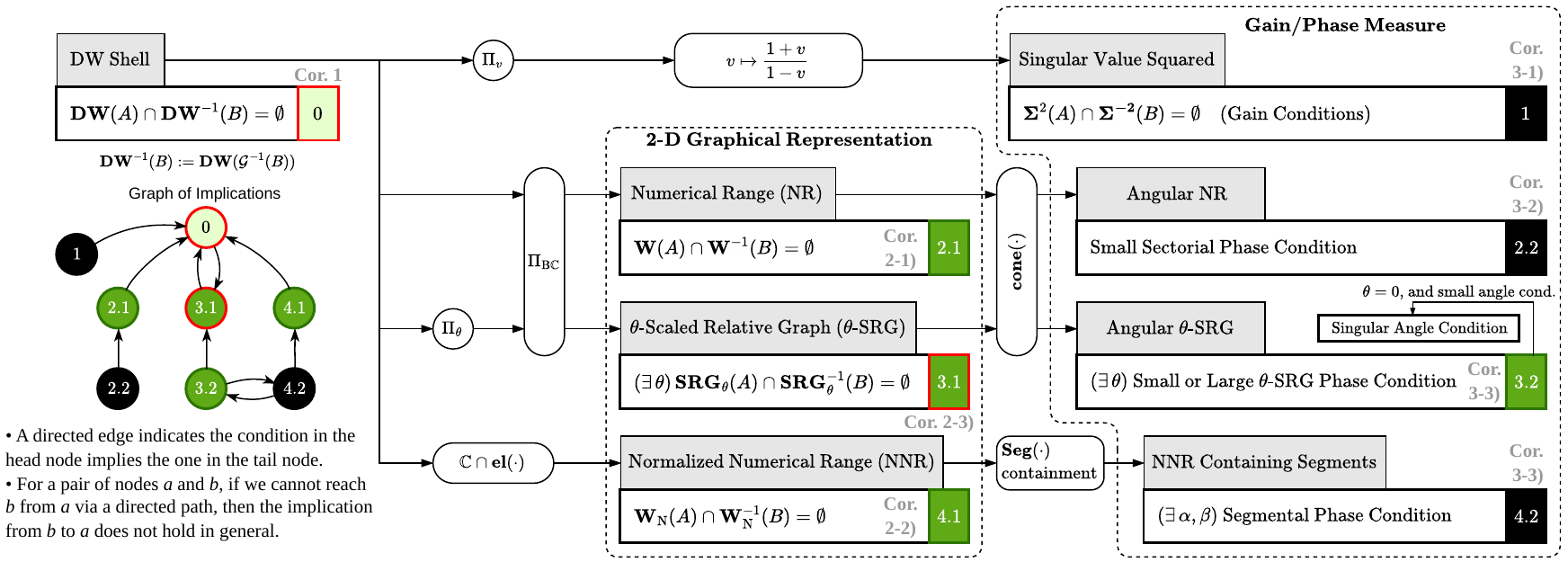}
    \caption{A DW perspective on various graphical representations\slash measures, separation conditions and their relationships: gray boxes depict the graphical representation of a standalone matrix. White labeled boxes (and the corresponding nodes in the graph of implications) represent separation conditions derived from the corresponding graphical representations. Existing conditions are labeled in black, and newly derived ones are in green. The least conservative conditions are highlighted with red circles.\label{fig:all-in-one-diagram}}
    \vspace{-.5cm}
\end{figure*} %

\subsection{Numerical visualization} \label{sec:dw-visual}
    Finally, we briefly discuss the numerical visualization of $\dwshell(\subsp[A])$. 
    The algorithm in \cite[Sec.~4]{lestasLargeScaleHeterogeneous2012} can be adapted directly to the subspace formulation using (\ref{defn-dw-r3}).
        We can find a basis pair $[D;N]$ for a given $\subsp[A]$ and normalize it as
        $\tilde{N}:= N(N\ct N + D\ct D)^{-1/2}, \tilde{D} = D(N\ct N + D\ct D)^{-1/2}$ so that 
        $X := [\tilde{D}\ct\hspace{.3em}\tilde{N}\ct]\ct$ is isometric and its columns still span $\subsp[A]$.
        Then 
        the real form (\ref{defn-dw-r3}) of the DW shell can be put in the following triple Hermitian form:
        \begin{align*}
            \{(\xi\ct X\ct E_1 X \xi, \xi\ct X\ct E_2 X \xi, \xi\ct X\ct E_3 X \xi ): \xi\ct\xi = 1\}. \numberthis \label{dw:triple-herm}
        \end{align*}
        For a unit vector $q \in \rF^3$, consider the Hermitian matrix $E(q):= \sum_{i=1}^3 q_i X\ct E_i X$.
        The smallest and largest eigenvalues of $E(q)$ give the two supporting planes of $\dwshell(\subsp[A])$ with normal vector $q$, and the associated normalized eigenvectors yield the tangent point on the boundary of the shell.
        Sampling $q$ on the unit sphere therefore generates boundary points and supporting half-spaces of $\dwshell(\subsp[A])$. 
        Then the convex hull of these points and the intersection of these half-spaces give the inner and outer approximations of $\convhull(\dwshell(\subsp[A]))$, respectively.
}}

\section{Phantom(s) of the DW Shell: Connections and Implications} \label{sec:dw-phantom}

{{\ignorespaces
\Cref{fig:all-in-one-diagram} summarizes the main results of \cref{sec:dw-phantom,sec:sep-cond}: gray boxes show graphical sets and gain and phase measures derived from them, while white boxes show the associated separation conditions. This section develops the sets and measures; the next examines the separation conditions and their implications.

We use matrix-based definitions and conditions in \cref{tab:2d-set-defns,fig:all-in-one-diagram} to provide a common setting for the cited results, whose extensions to a general setting differ across works. 
We will extend each graphical set and measure to subspaces through the map relating its matrix version to the matrix DW shell, applying that exact map to $\dwshell(\subsp[A])$.
Inverse sets, as part of the separation conditions in \cref{fig:all-in-one-diagram}, are forward sets evaluated on inverse graphs of matrices, which will be replaced by inverse subspaces to obtain the subspace counterparts.
Most implications between separation conditions remain valid for subspaces. However, Conds.~4.1 and 2.2 require additional assumptions to ensure separation of the associated 2-D sets and DW shells. We address these subtleties in \cref{sec:sep-cond}.
}}
\subsection{2-D graphical representations: NR, NNR, $\theta$-SRG}\label{subsec:shadow-interpret}

\subsubsection{Shadow interpretation}
We {{summarize}} the definitions of {{input-output}} gain and, to the best of our knowledge, all {{existing}} 2-D sets used in graphical system analysis in \cref{tab:2d-set-defns}. 
These sets can all be {{interpreted}}, in a certain sense, as `shadows' of the DW shell{{, allowing many of their properties to be derived from its geometry. We focus on the 2-D sets in this subsection and discuss gain together with phase in the next.}}
{{For $\theta\in\rF$, define}} the following {{set-valued map}} $\proj[\theta] : \ball[1]\to\sbound(\ball[1])$ {{by}}
\begin{align*}
    {{(w,v) \mapsto (e^{\ii*\theta}\left[\tRe(e^{-\ii*\theta}w) \pm \ii* \sqrt{1-v^2 - \tRe(e^{-\ii*\theta}w)^2}\right], v)}}.
\end{align*}
{{Then we have the following subspace generalization of the sets in \cref{tab:2d-set-defns}:}}
\begin{defn} \label{prop:dw-low-dim-set}
    {{Let $\subsp[A]\in \grasm(n,2n)$. We define its numerical range and the normalized version as:}}
    \begin{align*}
        \numran({{\subsp[A]}}) := \projdn\,\dwshell({{\subsp[A]}}),
        \nnr({{\subsp[A]}})  := \cF\cap \parahull(\dwshell({{\subsp[A]}})).
    \end{align*}
    {{For $\theta \in \rF$, we define 
        \begin{align*}
            \srg[\theta]({{\subsp[A]}}) := (\projdn\circ\proj[\theta])\,\dwshell({{\subsp[A]}}),
        \end{align*}
         referred to as the $\theta$-SRG}}.
\end{defn}

{{\ignorespaces 
\begin{fact}
    If $\subsp[A]=\gph(A)$, the graphical sets defined in \cref{prop:dw-low-dim-set} coincide with their matrix counterparts in \cref{tab:2d-set-defns}.
    In particular, $\srg[0](\subsp[A])=\srg(A)$.
\end{fact}

The geometric interpretations of the maps in \cref{prop:dw-low-dim-set} provide a handful of insights.
    We illustrate them in \cref{fig:optical-config}.

    To interpret $\projdn$, place a point light source at the north pole. The rays through the DW shell form a cone with its apex at the north pole, whose intersection with the complex plane gives the NR (\cref{subfig:dw-nr}).
    In addition, note that the north pole maps to $\infty$ under $\projdn$. If  $\infty$ is a normal spectral point of the subspace, its numerical range can either be a singleton $\{\infty\}$ if the subspace is $\invgph(\mato)$, or it can consist of a compact and convex set with the isolated point $\infty$ (in this case, the DW shell has a sharp point at the north pole). Otherwise, it would be the entire extended complex plane.

    Now let the light follow elliptic arcs joining the north and south poles through the nonpolar points of the DW shell (\cref{subfig:dw-nnr}). The union of these arcs forms the elliptic hull of the DW shell, whose intersection with the complex plane gives the NNR. 
    It is then intuitive to observe that the NNR always lies in the closed unit disk, and its intersection with the unit circle stems from the intersection between the DW shell and the unit sphere, therefore consisting precisely of the directions of the nonzero finite spectral points of $\subsp[A]$ by Property 3) of \cref{lem:dw-shell-basics}.
    Note that the north and south poles remain fixed under $\parahull((\cdot))$. Thus, $\nnr(\gph(0))=\nnr(\invgph(0))=\emptyset$.
}}
\begin{figure}[t]
    \vspace{-.3cm}
    \centering
    \subfloat[{$\projga$ {{(defined in Sec.~\ref{sec:ga-ph})}} \label{fig:projga}}]{\makebox[.2\textwidth][c]{
            \adjincludegraphics[Clip={0.1\width} {.25\height} {.1\width} {0.1\height}, height=2.5cm]{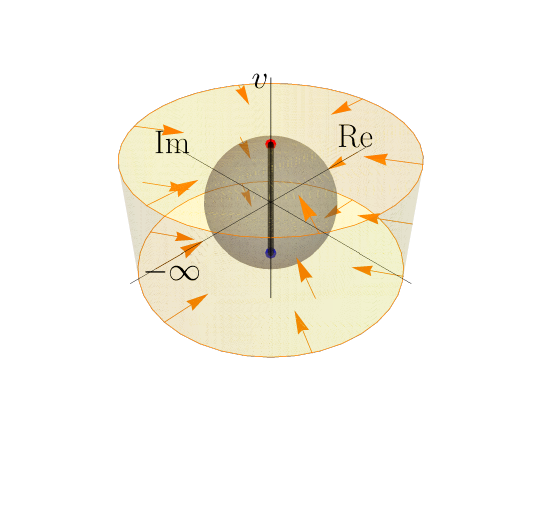}
        }
    }
    \subfloat[{$\projdn$}]{\makebox[.2\textwidth][c]{
            \adjincludegraphics[Clip={0.05\width} {.31\height} {.05\width} {0.04\height}, height=2.5cm]{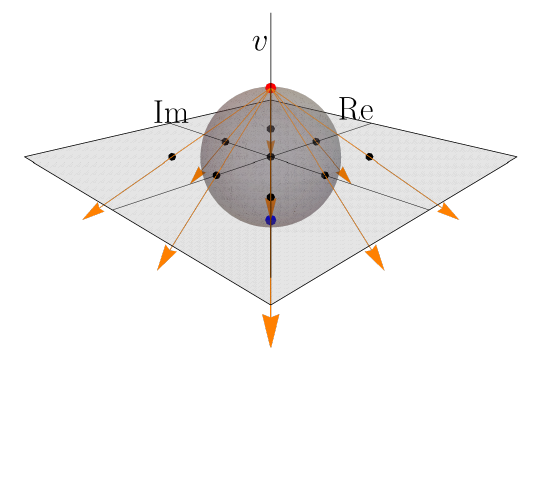}
        }
    }\\
    \subfloat[{$\proj[\theta]$\label{subfig:theta-proj}}]{\makebox[.15\textwidth][c]{
        \hspace{2em}
            \adjincludegraphics[Clip={0.06\width} {.2\height} {.05\width} {0.1\height}, height=2.5cm]{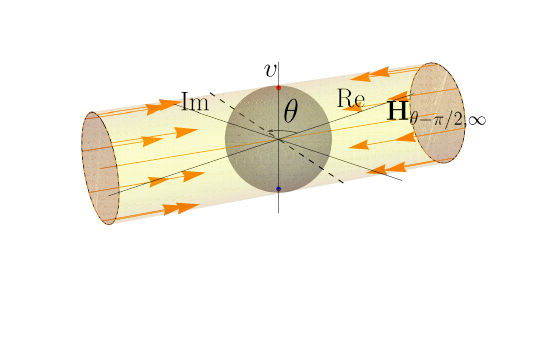}
        }
    }
    \subfloat[{The first projection in obtaining SSG \label{subfig:proj-ssg}}]{\makebox[.3\textwidth][c]{
            \adjincludegraphics[Clip={0.1\width} {.3\height} {.1\width} {0.1\height}, height=2.5cm]{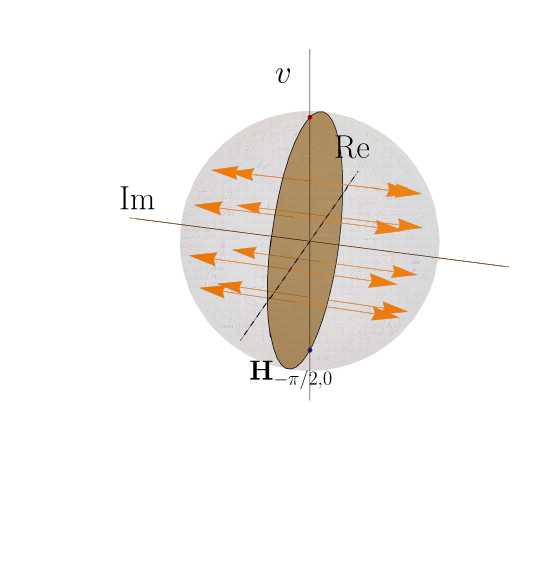}
        }
    }
    \caption{{{Interpretation of the maps that project}} the DW shell into the gain interval, numerical range, $\theta$-SRG, SSG.\label{fig:optical-config}} 
    \vspace{-.3cm}
\end{figure}
\begin{figure}[t]
    \centering\vspace{-.3cm}
    \subfloat[Numerical range (NR)\label{subfig:dw-nr}]{
        \parbox[c][2.5cm][c]{.19\textwidth}{
            \adjincludegraphics[
                Clip={0.3\width} {.3\height} {0.02\width} {0.02\height}, 
                height=2.5cm
            ]{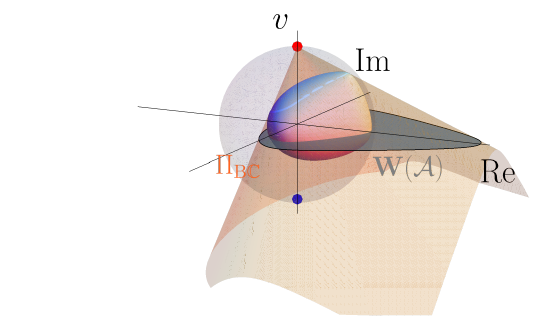}
        }
    }
    \subfloat[Normalized numerical range (NNR)\label{subfig:dw-nnr}]{
        \parbox[c][2.5cm][c]{.28\textwidth}{
            \centering
            \adjincludegraphics[
                Clip={0.1\width} {.3\height} {0.1\width} {0.02\height},
                height=2.5cm
            ]{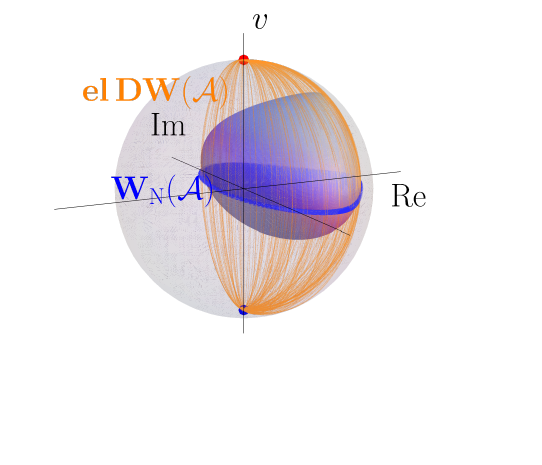}
        }
    }
    \\
    \subfloat[$\theta$-SRG\label{subfig:dw-srg}]{
        \parbox[c][3.3cm][c]{.22\textwidth}{
            \centering
            \adjincludegraphics[
                Clip={0.3\width} {.3\height} {0.02\width} {0.02\height},
                height=2.5cm
            ]{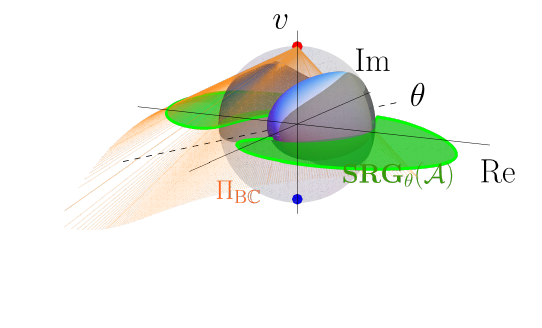}
        }
    }
    \subfloat[NR, NNR, and $\theta$-SRG\label{subfig:2d-in-one}]{
        \parbox[c][3.3cm][c]{.25\textwidth}{
            \centering
            \adjincludegraphics[
                Clip={0.0\width} {.0\height} {0.00\width} {0.00\height},
                height=3.3cm
            ]{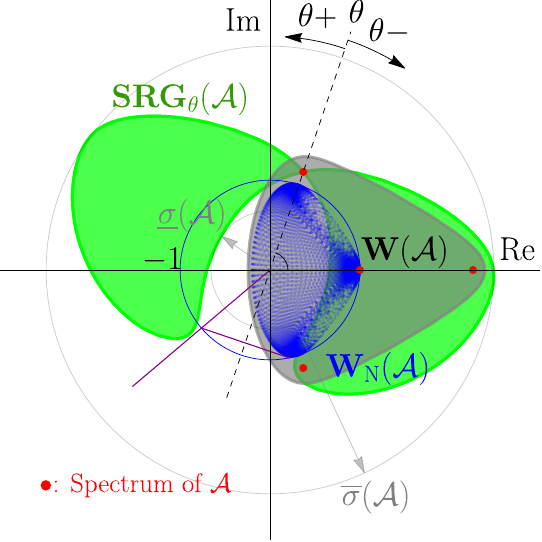}
        }
    }
    \caption{(a)-(c): Projecting a DW shell into NR, NNR, and $\theta$-SRG; (d): {{NR, NNR, and }}$\theta$-SRG of {{the graph of a}} real matrix $A$ (example drawn from \cite{patesScaledRelativeGraph2021}) {{in the complex plane}}.} \label{fig:allInOne}
    \vspace{-.7cm}
\end{figure}

{{We call the line $e^{\ii*\theta}\rF$ the $\theta$-axis.
    The positive and negative $\theta$-axes and the upper, lower, left, and right $\theta$-half-planes are obtained by rotating their standard counterparts ($\theta_0=0$) by $\theta$.}}
The {{$\theta$-SRG}} is obtained through a two-step projection of the DW shell {{(cf. \cref{subfig:theta-proj,subfig:dw-srg})}}: First, $\proj[\theta]$ projects {{each point of the DW shell, along the $(\theta+\frac{\pi}{2})$-axis, onto a pair of points on the unit sphere, symmetric about the plane spanned by the $\theta$- and $v$-axes}}; then {{these spherical pairs are projected}} to the complex plane {{by $\projdn$}}. 
The non-convexity of $\theta$-SRG stems from the first projection $\proj[\theta]$, which uses the {{spherical}} surface to capture the shadow. 
{{On a side note}}, the {{more recent}} SSG \cite{eijndenPhaseScaledGraphs2025} is obtained by modifying the first projection: {{ the $\mathrm{Re}$-$v$ plane divides the unit ball into two half-balls, and each point of the shell is projected only onto the spherical hemisphere bounding the half-ball containing it (\cref{subfig:proj-ssg}).}}

{{We next derive properties of the proposed $\theta$-SRGs through their relationship with the DW shell, exemplifying the usefulness of this unified geometric perspective.}}

\subsubsection{Basic properties of $\theta$-SRGs} \label{sec:theta-srg-basics}
Many properties of standard SRGs {{have been}} established {{through efforts across several works}} \cite{patesScaledRelativeGraph2021,ryuScaledRelativeGraphs2022,chaffeyGraphicalNonlinearSystem2023}
{{and properties of $\theta$-SRGs can be derived from these results}}. 
{{Alternatively, }}the geometry of DW shells outlined in \cref{sec:dw-shell-prelim}, together with the connections discussed in the previous section, provides {{a unified route to concise and intuitive development of these properties}}.
\begin{prop} \label{prop:theta-srg-basic}
    {{Let $\subsp[A] \in \grasm(n,2n)$}} and $\theta \in \rF$.
    \begin{enumerate}
        \item $\pi$-periodicity in $\theta$: $\srg[\theta]({{\subsp[A]}})= \srg[\theta+\pi]({{\subsp[A]}})$.
        \item Variation under {{subspace operations}}: 
        \begin{alignat*}{2}
            &\srg[\theta]({{(\mati_2 \otimes U)\subsp[A]}})        = \srg[\theta]({{\subsp[A]}})                         &&\hspace{-1.5em}\text{ for unitary }U,    \\
            &\srg[\theta]({{(\mati \oplus e^{\ii*\alpha} \mati)\subsp[A]}})  = e^{\ii*\alpha} \srg[\theta-\alpha]({{\subsp[A]}})  &&\hspace{-.2em}\text{ for }\alpha \in \rF,    \\
            &\srg[\theta]({{(\mati \oplus \gamma \mati) \subsp[A]}})         = \gamma \srg[\theta]({{\subsp[A]}})                  &&\hspace{-.2em}\text{ for }\gamma \in {{\poreals}}, \\
            &\srg[\theta]({{\cj(\subsp[A])}}) = \cj(\srg[-\theta]({{\subsp[A]}})). &&\\
            &\srg[\theta]({{\subsp[A]\inv}})  = \inv(\srg[-\theta]({{\subsp[A]}})), &&\\
            &\srg[\theta]({{\subsp[A]^\perp}}) = -(\cj(\srg[\theta]({{\subsp[A]}})))\inv. && 
        \end{alignat*} 
        \item Spectrum containment:
        The points of $\dwshell({{\subsp[A]}})$ {{which are contained in their own image under}} $\proj[\theta]$ give precisely the set {{$\projdn\inv \eigs(\subsp[A])$}}. Furthermore,
        $\srg[{{\theta}}]({{\subsp[A]}}) \cap e^{\ii*\theta}\rF = \eigs({{\subsp[A]}}) \cap e^{\ii*\theta}\rF$.
        {{Consequently}}, $\Lambda({{\subsp[A]}})\subseteq \srg[\theta]({{\subsp[A]}})$ always holds. Furthermore, we have
        \begin{align*}
            \Lambda({{\subsp[A]}}) = \cap_{\theta\in[0,\pi)} \srg[\theta]({{\subsp[A]}}).
        \end{align*} 
    \end{enumerate}
    {{As $\srg[\theta](\subsp[A])$ is symmetric about the $\theta$-axis, let $\srg[\theta+](\subsp[A])$ denote its portion in the closed upper $\theta$-half-plane:}}
    \begin{enumerate}
        \setcounter{enumi}{3}
        \item The set $\srg[\theta+]({{\subsp[A]}})$ is simply connected. For any two distinct points ${{w_1, w_2}} \in \srg[\theta+]({{\subsp[A]}})$, the minor (circular) arc centered on the $\theta$-axis that joins {{$w_1$}} and {{$w_2$}} lies entirely within $\srg[\theta+]({{\subsp[A]}})$.
        \item When ${{\subsp[A]}}$ is normal, the boundary $\sbound(\srg[\theta+]({{\subsp[A]}}))$ consists of a set of generalized arcs centered on the $\theta$-axis, whose endpoints are {{spectral points}} of ${{\subsp[A]}}$ {{or their reflections across the $\theta$-axis, taken in the closed upper $\theta$-half-plane}}.
    \end{enumerate}
\end{prop}
\begin{figure}[b]
    \vspace{-.7cm}
    \centering
    \subfloat[{{{$\projdn\inv \{w, w\inv,\cj(w)\inv\}$ }}\label{subfig:sph-pt-inv}}]{\makebox[.19\textwidth][c]{
            \adjincludegraphics[Clip={0.15\width} {.25\height} {0.08\width} {0.1\height}, height=2.3cm]{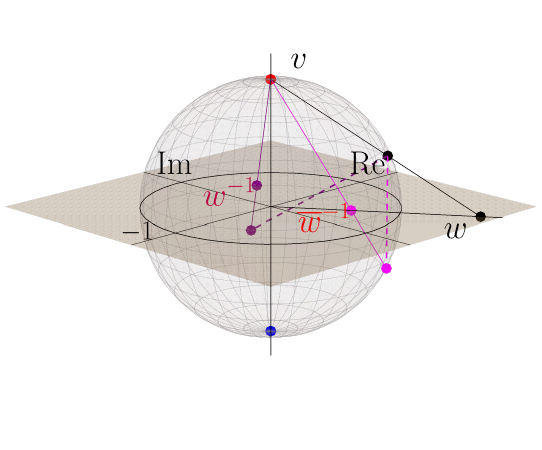}
        }
    }
    \subfloat[{$\projdn$ {{preserves the generalized circle}}\label{subfig:proj-circ}}]{\makebox[.28\textwidth][c]{
            \adjincludegraphics[Clip={0.2\width} {.26\height} {.17\width} {0.1\height}, height=2.3cm]{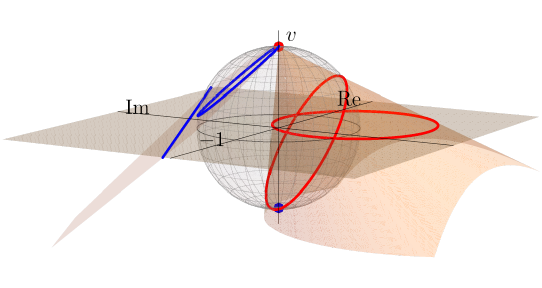}
        }
    }
    \caption{{{Geometry of $\projdn$ (restricted to the unit sphere)}}.\label{fig:projdn-demo}}
\end{figure}
\begin{proof}
{{\ignorespaces
1): Either the algebraic definition or the geometric interpretation of $\proj[\theta]$
(see \cref{subfig:theta-proj}) gives $\proj[\theta]=\proj[\theta+\pi]$. Then the $\pi$-periodicity follows from the definition in \cref{prop:dw-low-dim-set}.

2): We combine the DW-shell properties in
\cref{lem:dw-shell-basics} with those of $\proj[\theta]$ and $\projdn$.

The \emph{first equality} follows from the invariance of the DW shell under simultaneous block unitary transformations.
  
For the \emph{second equality}, $\dwshell((\mati\oplus e^{\ii*\alpha}\mati)\subsp[A])$ is obtained by rotating $\dwshell(\subsp[A])$ through $\alpha$ around the $v$ axis. 
        Projecting the rotated shell with $\projdn\circ\proj[\theta]$ is equivalent to projecting the original shell with $\projdn\circ\proj[\theta-\alpha]$ and then rotating the resulting planar set through $\alpha$ around the origin. Thus, the resulting set is $e^{\ii*\alpha}\srg[\theta-\alpha](\subsp[A])$.

        The \emph{third equality} follows from a direct verification.

To show \emph{the last three equalities}, recall that,
    \begin{align*}
        \dwshell(\cj(\subsp[A]))    &= (1\oplus -1 \oplus 1) \dwshell(\subsp[A]), \\
        \dwshell(\subsp[A]\inv)     &= (1\oplus -1 \oplus -1)  \dwshell(\subsp[A]), \\
        \dwshell(\subsp[A]^\perp)   &= (-1\oplus -1 \oplus -1) \dwshell(\subsp[A]).
    \end{align*}
    Due to the symmetry of $\dwshell(\cj(\subsp[A]))$ and $\dwshell(\subsp[A])$ about the $\mathrm{Re}$-$v$ plane, it follows that $\proj[-\theta] \dwshell(\subsp[A])$ and $\proj[\theta] \dwshell(\cj(\subsp[A]))$ are symmetric about the $\mathrm{Re}$-$v$ plane and 
    \begin{align*}(\projdn\circ\proj[\theta])\dwshell(\cj(\subsp[A])) = \cj((\projdn\circ\proj[-\theta])\dwshell(\subsp[A])),\end{align*} 
    i.e., the fourth equality holds.
    Further note that $\dwshell(\subsp[A]\inv) = (1\oplus 1 \oplus -1) \dwshell(\cj(\subsp[A]))$, meaning these two shells are symmetric about the complex plane. 
    Then $\proj[\theta] \dwshell(\subsp[A]\inv)$ and $\proj[\theta] \dwshell(\cj(\subsp[A]))$ are spherical sets symmetric about the complex plane. Furthermore, $\projdn$ maps any pair of spherical points symmetric about the complex plane to an inverse pair with respect to the unit circle. Thus, we have 
    $(\projdn\circ\proj[\theta])\dwshell(\subsp[A]\inv)  = \cj((\projdn\circ\proj[\theta])\dwshell(\cj(\subsp[A])))\inv$, where the right-hand side equals 
        $\inv((\projdn\circ\proj[-\theta])\dwshell(\subsp[A]) )$ from a few lines above. This shows that the fifth equality holds.
    For the last equality, note that $\dwshell(\subsp[A]^\perp)= (-1 \oplus -1 \oplus 1) (1 \oplus -1 \oplus 1) (1 \oplus -1 \oplus -1)\dwshell(\subsp[A])$, which implies
        $\dwshell(\subsp[A]^\perp) = \dwshell((\mati\oplus e^{\ii*\pi} \mati) \cj(\subsp[A]\inv))$. Apply $\projdn\circ \proj[\theta]$ to both sides and we have
        \begin{align*}
            \hspace{-1em}\srg[\theta](\subsp[A]^\perp) &= \srg[\theta]((\mati\oplus e^{\ii*\pi} \mati) \cj(\subsp[A]\inv)) = - \srg[\theta](\cj(\subsp[A]\inv))\\
            &=- \cj(\srg[-\theta](\subsp[A]\inv)) = - \cj(\srg[\theta](\subsp[A]))\inv.
        \end{align*}

3): Restrict $\projdn$ to the unit sphere, on which it is a bijection onto $\cF\cup\{\infty\}$. Note that $\proj[\theta]$ pushes every interior point of the unit ball to a pair of spherical points symmetric about the $\theta$-$v$ plane. Each point on the unit sphere is part of its own image pair under $\proj[\theta]$ and in particular, each point on the unit circle in the $\theta$-$v$ plane (i.e., $\projdn\inv\,e^{\ii*\theta}\rF \cup \{\infty\}$ ) is invariant under $\proj[\theta]$.
Then the first statement translates into $\dwshell(\subsp[A])\cap\sbound(\ball[1]) = \projdn\inv\eigs(\subsp[A])$, which follows from 3) of \cref{lem:dw-shell-basics}.
We also have:
\begin{align*}
    \projdn\inv\,(\eigs(\subsp[A])\cap e^{\ii*\theta}\rF) &= \projdn\inv \eigs(\subsp[A]) \cap \projdn\inv e^{\ii*\theta}\rF \\
    &= (\dwshell(\subsp[A]) \cap \sbound(\ball[1])) \cap \projdn\inv\,e^{\ii*\theta}\rF \\
    &= (\proj[\theta] \dwshell(\subsp[A])) \cap \projdn\inv\,e^{\ii*\theta}\rF.
\end{align*} Applying $\projdn$ to both sides gives: $\eigs(\subsp[A])\cap e^{\ii*\theta}\rF = \srg[\theta](\subsp[A])\cap e^{\ii*\theta}\rF$.
For the spectrum containment, note that
    $\projdn\projdn\inv \eigs(\subsp[A])\subseteq \projdn (\proj[\theta] \projdn\inv \eigs(\subsp[A])) = \projdn(\proj[\theta] \dwshell(\subsp[A])\cap \sbound(\ball[1])) \subseteq
     (\projdn\circ\proj[\theta]) \dwshell(\subsp[A]) = \srg[\theta](\subsp[A])$.

     Since $\eigs(\subsp[A])\subseteq \srg[\theta]({{\subsp[A]}})$ for all $\theta \in \rF$, it follows that $\eigs({{\subsp[A]}})\subseteq\cap_{\theta\in[0,\pi)}\srg[\theta]({{\subsp[A]}})$. For any $w\notin \Lambda({{\subsp[A]}})$, we have shown that $w\notin\srg[\angle w]({{\subsp[A]}})$ (where $\angle w$ can be set arbitrarily if $w=0$). Due to the $\pi$-periodicity, we can always find $\theta(w)\in[0,\pi)$ such that $\srg[\theta(w)]({{\subsp[A]}}) = \srg[\angle w]({{\subsp[A]}})$. This shows that $\cap_{\theta\in[0,\pi)} \srg[\theta]({{\subsp[A]}}) \subseteq \Lambda({{\subsp[A]}})$, which establishes the equality.
\begin{figure}[t]
    \centering \vspace{-.4cm}
    \adjincludegraphics[Clip = {.06\width} {0.2\height} {0} {0}, height=3cm]{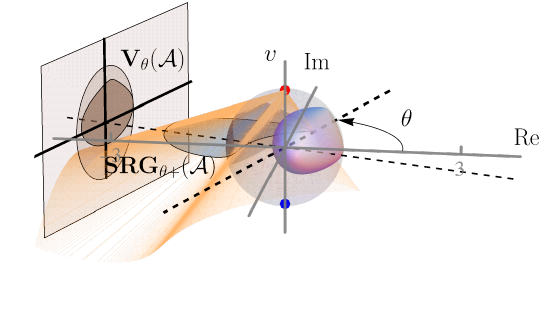}
    \caption{{$\theta$}-SRG, vertical numerical range, and the DW shell.\label{fig:vnumran}}
    \vspace{-.6cm}
\end{figure}

4) and 5): The projection $\projdn$ maps circles on the unit sphere to generalized circles in $\cF\cup\{\infty\}$. Specifically, as illustrated in \cref{subfig:proj-circ}, a hyperplane intersects the sphere either at a single point or along a circle. A point maps to a point, with the north pole mapped to $\infty$. A circle maps to a straight line if it passes through the north pole, and to a finite circle otherwise. 
    Moreover, if the normal to the hyperplane lies within the $\theta$-$v$ plane, its intersection with the sphere maps either to a circle centered on the $\theta$-axis or a line orthogonal to that axis. 
   To prove 4), consider the orthogonal projection of the DW shell onto the $\theta$-$v$ plane (see \cref{subfig:theta-proj,fig:vnumran}), which is convex%
   \footnote{
        {{From (\ref{dw:triple-herm}), the projection coincides in shape with the numerical range of $\vnumran[\theta](\subsp[A]):=\cos\theta X\ct E_1 X + \sin\theta X\ct E_2 X + \ii* X\ct E_3 X$. This should coincide with the numerical range that Pates \cite{patesScaledRelativeGraph2021} used to connect with the standard SRG through the Beltrami-Klein mapping (in which case, we set $\theta=0$).}}
    }. 
    The segment joining any two points in this projection therefore remains within it. Lifting this segment to the hemisphere corresponding to the upper $\theta$-half-plane (i.e., taking half of the image under $\proj[\theta]$) and then applying $\projdn$ yields the joining circular arc, or line segment, contained in $\srg[\theta+](\subsp[A])$. 
    Finally, for a normal subspace, the DW shell is the convex hull of its spectral points by \cref{lem:dw-shell-basics}-4). Its orthogonal projection onto the $\theta$-$\nu$ plane is therefore a convex polygon, whose edges yield the boundary arcs in 5) under the same lifting and stereographic projection.
}}  
\end{proof}

{{\ignorespaces
As illustrated in \cref{subfig:proj-circ}, a hyperplane cutting the unit sphere along a circle divides the unit ball into two portions and the sphere into two caps.
Choose $\theta$ so that the normal to the hyperplane lies in the $\theta$-$v$ plane. Then $\proj[\theta]$ maps each portion of the ball onto the corresponding spherical cap, and $\projdn$ further maps these caps onto the two regions divided by the image generalized circle: these could be the interior and exterior of a circle centered on the $\theta$-axis, two
half-planes if the image is a line perpendicular to the $\theta$-axis, or empty set versus the entire complex plane deleting a single point if the image is a singleton on the $\theta$-axis.
These projection properties connect the convexity of planar projections of the DW shell to the hyperbolic convexity of $\theta$-SRGs:
As the projection of the DW shell onto the $\theta$-$v$ plane is the intersection of its supporting closed half-planes, the $\theta$-SRG is the intersection of all of its closed bounding discs and the exterior to its open excluding discs, with both types of discs centered on the $\theta$-axis. 
}}

\subsection{Gain and phases} \label{sec:ga-ph}
{{Just as the SISO gain and phase can be read from the Nyquist plot, analogous measures can be extracted from the graphical representations discussed above.
A general approach is to choose a family of geometric sets parameterized by gain- or phase-type bounds and determine the tightest bounds for which the corresponding set contains the graphical representation.
Rather than introducing new measures, we focus in this section on interpreting and extending existing ones \cite{wangPhasesSemiSectorialMatrix2023,chenPhaseTheoryMultiinput2024,chenCyclicSmallPhase2025,baron-pradaMixedSmallGain2025}, within a unified geometric framework.
This clarifies their connections and distinctions. 
}}

\begin{figure}[t]
    \vspace{-.3cm}
    \centering
    \subfloat[Gain \label{fig:dw-gain}]{%
        \parbox[c][.34\columnwidth][t]{.3\columnwidth}{%
            \centering
            \adjincludegraphics[trim={0bp 82bp 65bp 1bp},clip,height=3cm]{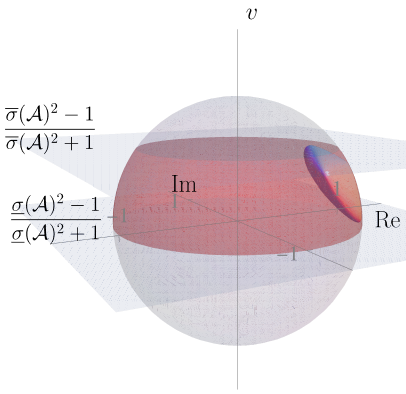}%
        }%
    }\hfill
    \subfloat[{Sectorial phase \label{fig:dw-sec-ph}}]{%
        \parbox[c][.34\columnwidth][t]{.3\columnwidth}{%
            \centering
            \adjincludegraphics[trim={35bp 40bp 50bp 1bp},clip,height=3cm]{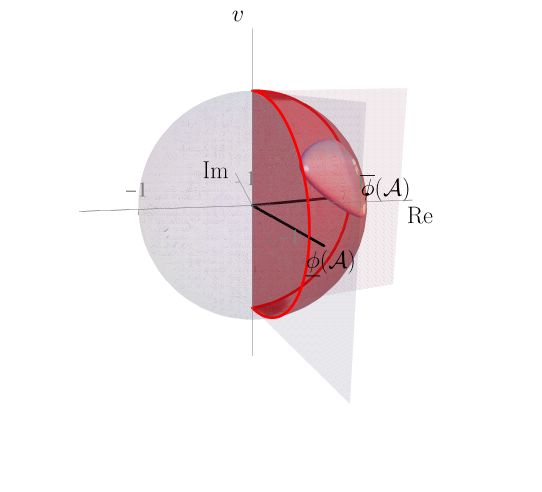}%
        }%
    }\hfill
    \subfloat[{$\theta$-angle ($\gamma$-seg. phase)}\label{fig:dw-seg-ph}]{%
        \parbox[c][.34\columnwidth][t]{.4\columnwidth}{%
            \centering
            \adjincludegraphics[trim={13bp 40bp 64bp 1bp},clip,height=3cm]{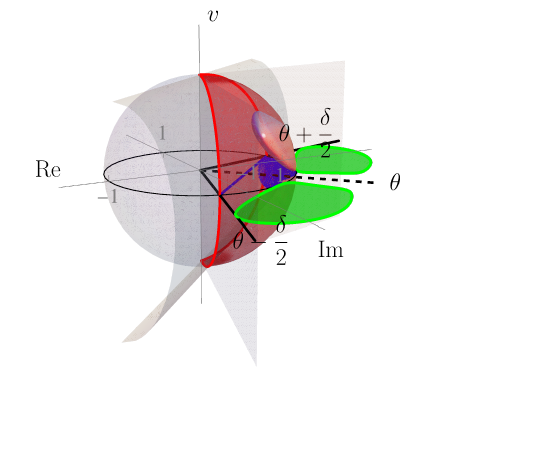}%
        }%
    }
    \caption{3-D regions {{enclosing}} the DW shell, {{defined by}} the gain, sectorial phase, and $\theta$-angle intervals.\label{fig:sep-3d-region}}
    \vspace{-.65cm}
\end{figure}

\subsubsection{Singular-valued based gain}
{{The input-output gain is developed somewhat independently of any graphical host. Nonetheless, it also admits a simple geometric interpretation through the DW shell.}}
{{Define the following maps:}}
    \begin{alignat*}{3}
        \projga: &\quad \ball[1] \to[-1,1],  &(w,v) &\mapsto v, \\
        f:       &\quad [-1,1]   \to \nnreals\cup\{\infty\}, & v &\mapsto \frac{1+v}{1-v}.
    \end{alignat*}
{{As illustrated in \cref{fig:dw-gain,fig:projga}, the $v$ coordinate of a DW shell encodes a normalized input-output gain. 
Specifically, the squared input-output gain interval can be obtained as: }}
    \begin{align}
        \gasq({{\subsp[A]}}) = (f\circ\projga)\,\dwshell({{\subsp[A]}}). \label{eq:dw-ga}
    \end{align} 
{{A horizontal hyperplane $\plane[v]$ with $v \in [-1,1]$ cuts the unit ball into two parts, with the upper (lower) part having a gain higher (smaller) than $\sqrt{\frac{1+v}{1-v}}$. The preimage of a squared gain interval under $f\circ\projga$ is precisely a `latitude'-bounded area of the unit ball where the DW shell resides, as shown in \cref{fig:dw-gain}.}}
{{In addition, by the second equality in \cref{prop:theta-srg-basic}-2), we have 
    $\srg[\theta](\subsp[A]) = e^{\ii*\theta} \srg[0]((\mati\oplus e^{-\ii*\theta}\mati)\subsp[A])$. 
    Then it is evident from the definition in \cref{tab:2d-set-defns} that the input-output gain interval is also encoded as the moduli of points in $\srg[\theta](\subsp[A])$.
    Specifically, the gain interval can be identified through the tightest bounding annulus, centered at the origin, of $\theta$-SRG for an arbitrary $\theta$.
        The preimage of such an annulus in the $\theta$-SRG plane under $\projdn\circ\proj[\theta]$ again gives the DW bounding region shown in \cref{fig:dw-gain}.
}}

{{\ignorespaces

\subsubsection{Phases and phase cones}
Parallel to the gain, different notions of multivariable phases have emerged, derived through different 2-D sets introduced in \cref{subsec:shadow-interpret}.

    {\bfseries \emph{Sectorial phases}}  \cite{wangPhasesSemiSectorialMatrix2023,chenPhaseTheoryMultiinput2024}: 
        They are defined through congruence canonical forms for semisectorial matrices, whose NRs do not contain the origin in their interiors.
        For a subspace $\subsp[A]$, the corresponding geometric description uses the conic hull $\conihull(\numran(\subsp[A]))$, with convexity inherited from $\numran(\subsp[A])$.
        When this set has well-defined bounding rays, their angles give the smallest and largest sectorial phases, $\phmin(\subsp[A])$ and $\phmax(\subsp[A])$, with $0\leqslant\phmax(\subsp[A])-\phmin(\subsp[A]) \leqslant\pi$.
            The preimage of an NR-bounding cone specified by a phase interval under $\projdn$ is precisely the DW-shell-bounding spherical wedge specified by the same phase interval (`longitude'-bounds), as illustrated in \cref{fig:dw-sec-ph}. 
    
    {\bfseries \emph{($\gamma$-)Segmental phase}} \cite[for $\gamma=0$]{chenSingularAngleNonlinear2025}, \cite{chenCyclicSmallPhase2025}: 
        This notion is based on the NNR, which lies in the closed unit disk.
        For a fixed bisector angle $\gamma$, consider circular segments of the unit disk with central angle $\delta\in[0,2\pi]$.
        The smallest such segment containing the NNR determines the $\gamma$-segmental phase interval, with endpoints $\gamma-\delta/2$ and $\gamma+\delta/2$.
        As depicted in \cref{fig:dw-seg-ph} (NNR is shown in blue), the preimage of an NNR-bounding segment specified by a phase interval under the operation $\cF\cap \parahull({(\cdot)})$ is the DW-shell-bounding spherical segment specified by the same phase interval.
    
    {\bfseries \emph{($\theta$-)SRG angle}} \cite[for $\theta=0$]{baron-pradaMixedSmallGain2025}:
        The conic hull of the $\theta$-SRG also derives a phase-type notion. 
            For a subspace $\subsp[A]$ and a $\theta \in \rF$, due to \cref{lem:dw-shell-basics}-4), $0$ is never an interior point of $\srg[\theta](\subsp[A])$ as long as $\eigs(\subsp[A])$ is a finite set. 
            The conic hull $\conihull(\srg[\theta](\subsp[A]))$ consists of two convex cones symmetric about the $\theta$-axis (cf. \cref{fig:dw-seg-ph,fig:vnumran,subfig:2d-in-one}). 
            When we focus on the convex half in the upper $\theta$-half-plane, its smallest and largest absolute angular deviations from the \emph{positive}-$\theta$-axis define what we refer to as the \emph{$\theta$-angles}, denoted $\angmin[\theta](\subsp[A]),\angmax[\theta](\subsp[A])$.  

            The $\theta$-angles, when defined, always lie within $[0,\pi]$. We discuss several extreme cases:
            \begin{itemize}
                \item  $\theta$-angles are undefined for $\gph(0)$ and $\invgph(0)$, whose DW shell consists solely of the south or the north pole and whose $\theta$-SRG is $\{0\}$ or $\{\infty\}$ for all $\theta$.
                \item  If $\eigs(\subsp[A])$ is infinite, then $\proj[\theta]\dwshell(\subsp[A])$ is the entire unit sphere, so $\srg[\theta](\subsp[A])$ is the entire extended complex plane. In this case, $\angmin[\theta](\subsp[A])=0$ and $\angmax[\theta](\subsp[A])=\pi$ for all $\theta$.
                \item \cref{prop:theta-srg-basic}-3) also implies that $\angmin[\theta](\subsp[A])=0$ if $\subsp[A]$ has an eigenvalue on the positive $\theta$-axis, and $\angmax[\theta](\subsp[A])=\pi$ if it has one on the negative $\theta$-axis. The converses need not hold. For example, the DW shell of $\gph(\left[\begin{smallmatrix}
                    0 & 1\\ 0 & 0
                \end{smallmatrix}\right])$ is an ellipsoid sitting on the south pole. Its $\theta$-SRG is always two discs of radius $\frac{1}{2}$, centered at $\frac{1}{2}e^{\ii*(\theta\pm \pi/2)}$. Thus it has $\theta$-angles $0$ and $\pi$ for every $\theta$, although its spectrum is $\{0\}$.
            \end{itemize}
           
            Similar to the gain identity $\svmin(\subsp[A]\inv)=1/\svmax(\subsp[A])$, the $\theta$-angles satisfy the following reflection identity whenever defined:
}}
            \begin{align}
                \label{eq:add-reflex}
                \angmin[\theta]((\mati \oplus -\mati){{\subsp[A]}}) &= \angmin[\theta+\pi]({{\subsp[A]}})=  \pi - \angmax[\theta]({{\subsp[A]}}),
            \end{align} which follows from \cref{prop:theta-srg-basic}-1) and 2).
            {{The second equality indicates that the smallest and largest $\theta$-angles contain redundant information when $\theta$ is considered over a $2\pi$ interval. Thus, it suffices to consider either both angle bounds over a $\pi$ interval or only the largest angle $\angmax[\theta](\subsp[A])$ over a $2\pi$ interval. }}
{{\ignorespaces

    The cone defined by the upper $\theta$-angle bound corresponds
    to a segmental region enclosing the DW shell (see \cref{fig:dw-seg-ph}). This is precisely the region associated with the $\gamma$-segmental phase interval when $\gamma=\theta$ and $\delta=2\,\angmax[\theta](\subsp[A])$.
    Thus, the two notions encode the same angular bound, despite being derived from different graphical sets.
    We henceforth focus on the $\theta$-angle.
\subsubsection{Subtleties at the north and south poles} \label{sec:subtlety-n-s-pole}
    Special care is needed at the north and south poles, which correspond, under $\projdn$, to $\infty$ and $0$, respectively.
    As mentioned before, under $\parahull((\cdot))$, each pole remains fixed and therefore contributes no point to the equatorial cross-section defining the NNR. Consequently, disjoint NNRs do not rule out an intersection of their generating DW shells at either pole.

    Geometrically, a phase interval specifies a phase cone, but never informs whether $0$ or $\infty$ are included regardless of whether the interval is closed or open.
        This subtlety seems to make no difference in the matrix setting, but the system analysis always sees complications at the system poles and zeros and phase response definition often requires caution there. 
        In the next section, we clarify this issue through the geometric perspective presented here and show that a sectorial phase inequality alone does not guarantee general subspace separation, whereas a $\theta$-angle inequality does.
}}

\section{Separation Conditions and Their Relationships \label{sec:sep-cond}}
{{So far, we have discussed mostly graphical sets associated with a single subspace and established their connections. Our motivating problem~(\ref{eq:pt-gsep}), however, concerns the separation of two subspaces. In this section, we use these graphical tools to characterize it and establish a \emph{complete} implication diagram as illustrated in \cref{fig:all-in-one-diagram}, relating the resulting graphical conditions.}}
These conditions are represented as white boxes in \cref{fig:all-in-one-diagram} and labelled locally therein. Conditions labelled in green are not explicitly spelled out {{even for the matrix setting}} in the literature.
The relative conservatism of these conditions is illustrated by a graph of implications in \cref{fig:all-in-one-diagram}, where each node represents a condition and each directed edge indicates that the head condition implies the tail one (given any assumptions imposed by the head condition). Conditions enclosed in red frames are identified as the least restrictive.

\renewcommand{\arraystretch}{1.03}
\subsection{DW separation condition} \label{sec:dw-sep}
{{\ignorespaces
In the matrix setting, it is known that the separation of DW shells provides a necessary and sufficient condition for the nonsingularity along the unitary orbit around two matrices \cite{liEigenvaluesSumMatrices2008}. We generalize this result to the following subspace setting, which is central to our subsequent developments:
\begin{theorem}[DW Separation] \label{thm:dw-subspace}
    Let $[X_1;X_2], [Y_1;Y_2]$ be the respective basis pairs of $\mathcal{S},\mathcal{T} \in \grasm(n,2n)$. The following statements are equivalent:
    \begin{enumerate}
        \item $\dwshell(\mathcal{S}) \cap \dwshell(\mathcal{T}) = \emptyset$;
        \item ($(\mati_2 \otimes U)\mathcal{S})\cap \mathcal{T} = \{0\}$ for all unitary $U \in \cF^{n\times n}$;
        \item
        $ \begin{bNiceArray}{c@{\hspace{.4em}}c}
                U X_1 & Y_1 \\
                U X_2 & Y_2    
            \end{bNiceArray}$ is nonsingular for all unitary $U\in \cF^{n\times n}$.
    \end{enumerate}
\end{theorem}
\begin{proof}
    See Appendix~\ref{app:proof-dw-sep-ext}.
\end{proof}
}}
{{\ignorespaces
    For feedback interconnections, we often deal with the separation between $\subsp[A]$ and $\subsp[B]\inv$ (cf. the form of (\ref{eq:pt-gsep})).
        Then \cref{thm:dw-subspace} can be rephrased as:
        \begin{corol} \label{corol:dw-sep-fbck}
            Let $[D_{\subsp[A]};N_{\subsp[A]}]$ and $[D_{\subsp[B]};N_{\subsp[B]}]$ be the respective basis pairs of $\subsp[A], \subsp[B] \in \grasm(n,2n)$. The following statements are equivalent:
            \begin{enumerate}
                \item $\dwshell(\subsp[A])\cap\dwshell(\subsp[B]\inv)=\emptyset$;
                \item $((\mati_2 \otimes U)\subsp[A])\cap\subsp[B]\inv=\{0\}$ for all unitary $U\in\cF^{n\times n}$;
                \item 
                $\begin{bNiceArray}{c@{\hspace{.4em}}c}
                    U D_{\subsp[A]} & N_{\subsp[B]} \\
                    U N_{\subsp[A]} & D_{\subsp[B]}
                \end{bNiceArray}$
                is nonsingular for all unitary $U\in\cF^{n\times n}$.
            \end{enumerate}
        \end{corol}

    Suppose first that $\subsp[A]$ and $\subsp[B]$ have matrix representations $A$ and $B$, respectively.
    By \cref{corol:dw-sep-fbck},
        $\dwshell(A)\cap\dwshell(\invgph(B))=\emptyset$
    if and only if
    $\begin{bNiceArray}{c@{\hspace{.5em}}c}
        U & B\\
        UA & \mati
    \end{bNiceArray}$
    is nonsingular for every unitary $U$,
    which is equivalent to the nonsingularity of $\mati - U A U\ct B$ for all unitary $U$.
    In particular, all of them imply $\subsp[A]\cap\subsp[B]\inv=\{0\}$.
           
    Suppose instead that $\subsp[B]\inv$ has a matrix representation $C$, even though $\subsp[B]$ need not have one. The same theorem gives
        $\dwshell(A)\,\cap\,\dwshell(C)=\emptyset$
    if and only if
    $\begin{bNiceArray}{c@{\hspace{.5em}}c}
        U & \mati\\
        UA & C
    \end{bNiceArray}$
    is nonsingular for every unitary $U$, or equivalently, $C-U AU\ct$ is nonsingular for every
    unitary $U$, which recovers the result in \cite[Thm.~2.1]{liEigenvaluesSumMatrices2008}.  All of these again imply $\subsp[A]\cap\subsp[B]\inv=\{0\}$.

    The matrix-based DW formulation therefore suffices when each subspace or its inverse admits a matrix representation. The subspace formulation also covers cases where neither representation exists for a subspace, corresponding in the system setting to the existence of coincident poles and zeros.

    We shall use hereafter some separating hyperplane based arguments, which can be summarized as:
    \begin{prop} \label{prop:sep-hyper-plane}
        Let $\subsp[A], \subsp[B]\in \grasm(n,2n)$. Then:
        \begin{enumerate}
            \item $\dwshell(\subsp[A])\cap \dwshell(\subsp[B]\inv) = \emptyset$ if and only if there exists a hyperplane in $\cF\times \rF$ such that the two shells lie in \emph{opposite open} half spaces divided by the hyperplane;
            \item For $\theta \in \rF$, 
            $\srg[\theta](\subsp[A])\cap\srg[\theta](\subsp[B]\inv) = \emptyset$ if and only if there exists a finite circle centered on the $\theta$-axis such that one $\theta$-SRG lies strictly on the finite side of the circle and the other lies strictly on the infinite side.
        \end{enumerate}
    \end{prop} 
    The DW shell could be a non-convex ellipsoid surface for $n=2$. However, it always lies in the unit ball and contacts the unit sphere. Then nesting is impossible for the strict separation of two shells. Thus for all cases, the shells are separated if and only if their convex hulls, which are also compact, are separated. Then the first statement follows from the separation theorem \cite[Cor.~11.4.2]{rockafellarConvexAnalysis1970} applied to the compact convex hulls of the shells.
    The second, as discussed above \cref{sec:ga-ph}, can be viewed as a projected version of the first.
    \vspace{-.2cm}

    }}

\subsection{Formulas for {{the inverse subspace}}} \label{sec:inv-formula}
{{\ignorespaces
Since loop components are usually given in their forward representations, we provide explicit formulas expressing the graphical sets and associated gain and phase measures of an inverse subspace in terms of those of the original subspace, which is partly done in \cref{lem:dw-shell-basics,prop:theta-srg-basic}.

For $\subsp[B] \in \grasm(n,2n)$, recall from \cref{lem:dw-shell-basics,prop:theta-srg-basic} that
\begin{align*}
    \dwshell(\subsp[B]\inv) &= (1\oplus -1 \oplus -1) \dwshell(\subsp[B]),\\
    \srg[\theta](\subsp[B]\inv) &= (\srg[-\theta](\subsp[B]))\inv.
\end{align*}
It is straightforward to see from (\ref{eq:dw-ga}) and the form of $f$ that 
\begin{align*}
    \gasq(\subsp[B]\inv) =  \inv(\gasq(\subsp[B])) = \left[\svmax(\subsp[B])^{-2},\svmin(\subsp[B])^{-2}\right].
\end{align*} Also, it follows from the $\theta$-SRG equality that:
\begin{align*}
    \angmin[\theta](\subsp[B]\inv) = \angmin[-\theta](\subsp[B]), \quad \angmax[\theta](\subsp[B]\inv) = \angmax[-\theta](\subsp[B]).
    \numberthis \label{eq:ang-inv}
\end{align*}

Geometrically, the stereographic projection $\projdn$ takes the conic hull, with the apex at the north pole, of the set being mapped, and then takes the intersection between the resulting cone and the extended complex plane.
    Setting the apex of the cone to the south pole instead yields the map 
    \begin{align*}
        \projup:\;\ball[1]\to \cF\cup\{\infty\}, \quad (w,v)\mapsto \frac{w}{1+v}.
    \end{align*} 
Then the numerical range of the inverse subspace satisfies
    \begin{align*}
        \numran(\subsp[B]\inv)
        &= \projdn\,\dwshell(\subsp[B]\inv) \numberthis \label{eq:nr-inv}\\
        &= \projdn\bigl((1\oplus-\mati_2)
            \dwshell(\subsp[B])\bigr)
        = \cj(\projup\,\dwshell(\subsp[B])). 
    \end{align*}
It follows that the sectorial phase satisfies:
\begin{align}
    \minph(\subsp[B]\inv) = -\maxph(\subsp[B]),\quad \maxph(\subsp[B]\inv) = - \minph(\subsp[B]). \label{eq:sec-ph}
\end{align}

Finally, for the normalized numerical range, 
    $\nnr(\subsp[B]\inv)=\cj(\nnr(\subsp[B]))$. 
    This follows directly from the definition, or geometrically from the interpretation of $\parahull((\cdot))$ (see \cref{subfig:dw-nnr}) and the symmetry relating $\dwshell(\subsp[B]\inv)$ to $\dwshell(\subsp[B])$.

}}

\subsection{Separation theorems and their relationships} \label{subsec:sep-thms}

By simultaneously applying the {{maps introduced in \cref{prop:dw-low-dim-set}}} to $\dwshell({{\subsp[A]}})$ and $\dwshell({{\subsp[B]\inv}})$, we obtain a family of `projected' versions of the DW separation condition. 
Each of these projected conditions serves as a sufficient condition for the DW separation condition, {{but in the subspace setting there are subtleties unseen in the matrix setting at the two poles.}}

\subsubsection{2-D graphical conditions}
{{The NR and $\theta$-SRG are the subsets of the extended complex plane, with $\infty$ included whenever the corresponding DW shell contains the north pole. Consequently, two graphical sets of the same type that both contain $\infty$ intersect there and should be interpreted as not satisfying the separation requirement. On the 2-D graphical representations, we have:}}
\begin{corol} \label{prop:2d-sep}
    {{Let $\subsp[A], \subsp[B]\in \grasm(n,2n)$}}. Then {{$((\mati_2 \otimes U)\subsp[A] )\cap \subsp[B]\inv = \{0\}$}} for all unitary $U$ if any of the following holds:
    \begin{enumerate}
        \item $\numran(\subsp[A]) \cap \numran(\subsp[B]\inv) = \emptyset$;
        \item {{Each of the pairs $(\subsp[A], \subsp[B]\inv), (\subsp[A]\inv,\subsp[B])$ contains at least one subspace admitting a matrix representation, and}} $\nnr(\subsp[A]) \cap \nnr(\subsp[B]\inv) = \emptyset$;
        \item $\srg[\theta](\subsp[A]) \cap \srg[\theta](\subsp[B]\inv) = \emptyset$ for some $\theta \in [0,\pi)$.
    \end{enumerate}
    In particular, the `only if' side holds for condition 3).
\end{corol}
\begin{proof}
    {{\ignorespaces
        By \cref{corol:dw-sep-fbck}, sufficiency follows if each condition implies DW-shell separation.
        For 1), $\projdn$ is defined on the entire unit ball, and the numerical range is its image of the DW shell. Thus, disjoint numerical ranges imply disjoint DW shells. The same argument applies to 3): although $\projdn\circ\proj[\theta]$ is set-valued, every point in the unit ball has a nonempty image under it, so any point shared by the DW shells would produce a pair of points shared by their $\theta$-SRGs under $\projdn\circ\proj[\theta]$.

        It follows from the definition of $\cF\circ \parahull((\cdot))$ that disjoint normalized numerical range implies that $\dwshell(\subsp[A])$ and $\dwshell(\subsp[B]\inv)$ do not intersect on $\ball[1]\setminus \{(0,1),(0,-1)\}$. The additional matrix representation condition, by \cref{prop:n-s-pole}, guarantees that they do not intersect at the two poles.
    }}

     For the necessity of 3), if $\dwshell(\subsp[A])\,\cap\,\dwshell(\subsp[B]\inv)=\emptyset$, then a hyperplane, with its normal $(\eta,\nu)\in\cF\times\rF$, strictly separates the two shells. If $\eta=0$, the separating plane is horizontal, so the shells are separated by gain. Their $\theta$-SRGs are therefore separated by a circle centered at the origin for every $\theta$.
    Otherwise, choose $\theta=\arg\eta$ (mod $\pi$). As discussed below \cref{prop:theta-srg-basic}, the intersection between the separating hyperplane and the unit sphere is projected to a separating circle centered on the $\theta$-axis that strictly separates the $\theta$-SRGs. 
\end{proof}

\subsubsection{Gain and phase conditions}
{{\ignorespaces
    As illustrated in \cref{fig:sep-3d-region}, gain and phase bounds define regions in the unit ball that enclose the DW shells. Separation of these regions therefore can be used to guarantee DW-shell separation and can often be expressed through simple inequalities involving the gain and phase bounds. 
    Before presenting these conditions, we introduce a technical lemma that will help show the sufficiency of a clean $\theta$-angle condition for the DW-shell separation.
    \begin{lemma} \label{lem:angle-lem}
        Let $[D;N]$ be a basis pair of $\subsp[A] \in \grasm(n,2n)$ and assume $N, D\neq \mato$.  Then $\dwshell(\subsp[A])$ is not a singleton at either pole.
        If $\dwshell(\subsp[A])$ intersects the polar diameter $\{(0,v):-1\leqslant v \leqslant 1\}$, then $\angmin[\theta](\subsp[A]) \leqslant \pi/2 \leqslant \angmax[\theta](\subsp[A])$ for all $\theta$.
    \end{lemma}
    \begin{proof}
        Fix an arbitrary $\theta$. Under $\projdn\circ\proj[\theta]$, each nonpolar point of the polar diameter produces a pair of nonzero finite points whose absolute angular deviations from the positive $\theta$-axis are $\pi/2$. This establishes the inequalities if $\dwshell(\subsp[A])$ intersects a nonpolar point.
        
        Meanwhile, the north and south poles map to 0 and $\infty$, respectively.
        The nonzero assumption on $N, D$ ensures $\dwshell(\subsp[A])$ is not a singleton at either pole and has well-defined $\theta$-angles.
            Suppose the only intersection between $\dwshell(\subsp[A])$ and the polar diameter is the south pole, apart from which $\dwshell(\subsp[A])$ must contain a point \emph{outside} the polar diameter which maps to a nonzero finite point $w$ in $\srg[\theta+](\subsp[A])$.
            Then the line segment joining these two DW-shell points maps to a generalized circular arc joining $0$ and $w$, centered on the $\theta$-axis and lying in $\srg[\theta+](\subsp[A])$. 
            Since the generalized arc is tangent at the origin to a line perpendicular to the $\theta$-axis, the argument deviations from $\theta$ by the arc points near $0$ approach $\pi/2$. 
            If the arc degenerates to a line segment, these deviations equal $\pi/2$. Thus, both inequalities hold.
        The north-pole case follows by applying the same argument to $\subsp[A]\inv$ and invoking \cref{eq:ang-inv}.
    \end{proof}

    We next state the gain and phase conditions precisely and address the subtleties mentioned by \cref{sec:subtlety-n-s-pole}.
}}
\begin{corol} \label{prop:1d-sep}
    {{Let $\subsp[A], \subsp[B] \in \grasm(n,2n)$}}. Then {{$((\mati_2 \otimes U)\subsp[A]) \cap \subsp[B]\inv = \{0\}$}} for all unitary $U$ if any of the following holds:
    \begin{enumerate}
        \item \emph{Gain conditions}: $\svmax({{\subsp[A]}})\svmax({{\subsp[B]}})<1$, or $\svmin({{\subsp[A]}})\svmin({{\subsp[B]}})>1$;
        \item \emph{Sectorial phase condition}: Both subspaces have well-defined phases, there exists an integer $k$ such that
        \begin{align*}
            2k\pi  < \phmin({{\subsp[A]}})+\phmin({{\subsp[B]}}) \leqslant \phmax({{\subsp[A]}})+\phmax({{\subsp[B]}}) < 2(k+1)\pi. 
        \end{align*}
        {{In addition, each of the pairs $(\subsp[A],\subsp[B]\inv)$ and $(\subsp[A]\inv,\subsp[B])$ contains at least one subspace admitting a matrix representation. 
            If neither $\subsp[A]$ nor $\subsp[A]\inv$ admits a matrix representation, we further require $0\notin\numran(\subsp[B])$. The same requirement applies with $\subsp[A]$ and $\subsp[B]$ interchanged.}}
        \item \emph{$\theta$-angle conditions}:
        There exists $\theta\in [0,\pi)$ such that 
            $\angmax[\theta]({{\subsp[A]}}) < \angmin[-\theta]({{\subsp[B]}}),\text{ or }\angmin[\theta]({{\subsp[A]}}) > \angmax[-\theta]({{\subsp[B]}}).$
    \end{enumerate}
\end{corol}
{{\ignorespaces
\begin{proof}
    The two alternative gain inequalities characterize strict separation of $\dwshell(\subsp[A])$ and $\dwshell(\subsp[B]\inv)$ by a horizontal plane, and hence each is sufficient.

    By \cref{eq:sec-ph}, inequalities in 2) can be written as for some integer $k$,
        $\minph(\subsp[A])>\maxph(\subsp[B]\inv)+2k\pi$,
        $\maxph(\subsp[A])<\minph(\subsp[B]\inv)+2(k+1)\pi$, showing that the conic hulls of $\numran(\subsp[A])$ and $\numran(\subsp[B]\inv)$ have no common nonzero finite point. 
        Any intersection of their NRs must be at 0 or $\infty$, whose preimage under $\projdn$ is the entire polar diameter (dashed in \cref{fig:sec-ph-side}). 
        
    We next show that 2) guarantees the DW separation on the polar diameter.
    The pairwise matrix-representation assumptions exclude a common north or south pole. Suppose, for a contradiction, that the two shells share a nonpolar point on the polar diameter. 
    The phase inequalities imply that the sum of phase widths is strictly less than $2\pi$, so at least one component has phase width strictly less than $\pi$.
    Assume without loss of generality that this is the width of $\subsp[A]$. Let $[D;N]$ be a normalized basis pair of $\subsp[A]$. The shared nonpolar point corresponds to a nonzero $x$ such that $Nx\neq 0, Dx \neq 0, x\ct D\ct N x = 0$. Since the phase width is strictly less than $\pi$, there exist angles $\delta_1, \delta_2$ with $0<|\delta_1-\delta_2|<\pi$ such that both rotated shells $\dwshell((\mati\oplus e^{\ii*\delta_i}\mati)\subsp[A])$ lie in the closed right half-ball, i.e., $X_i\ct E_1 X_i$ in (\ref{dw:triple-herm}) is positive semidefinite where $X_i = (\mati \oplus e^{\ii*\delta_i} \mati) [D\ct\hspace{.3em}N\ct]\ct$.  
    Since $x\ct D\ct N x = 0$, we have $x\ct X_i\ct E_1 X_i x = 0 \Rightarrow x \in \nSp(X_i\ct E_1 X_i)$ for $i=1,2$, i.e.,
    \begin{align*}
        \hermp((e^{\ii*\delta_1} D\ct N))x = 0, \hspace{1em}
        \hermp((e^{\ii*\delta_2} D\ct N))x = 0. 
    \end{align*}
    Then $0<|\delta_1-\delta_2|<\pi$ implies that $N\ct Dx = D\ct Nx = 0$. Together with $Nx \neq 0, Dx \neq 0$, this implies that both $N$ and $D$ are singular, i.e., neither $\subsp[A]$ nor $\subsp[A]\inv$ admits a matrix representation. 
    Meanwhile, by (\ref{eq:nr-inv}) this nonpolar point in $\dwshell(\subsp[B]\inv)$ yields $0\in \numran(\subsp[B])$, which contradicts 2). Thus, under 2), $\dwshell(\subsp[A])$ and $\dwshell(\subsp[B]\inv)$ share no common points on the polar diameter. Then $2)\Rightarrow \dwshell(\subsp[A])\cap \dwshell(\subsp[B]\inv) = \emptyset$.

    For 3), the angle inequalities separate the nonzero finite points of $\srg[\theta](\subsp[A])$ and $\srg[\theta](\subsp[B]\inv)$, which further implies the separation of nonpolar points of $\dwshell(\subsp[A])$ and $\dwshell(\subsp[B]\inv)$. Any intersection of the two shells thus must occur at a pole.
    The presumption on well-defined $\theta$-angles implies both subspaces are not in $\{\gph(0),\invgph(0)\}$ and their DW shells are not a singleton at a pole.
    Now if the two shells share either pole, then \cref{lem:angle-lem} implies for all $\theta$,
    \begin{align*}
        \angmin[\theta](\subsp[A])
        \leqslant \pi/2
        \leqslant \angmax[\theta](\subsp[A]),
        \qquad
        \angmin[-\theta](\subsp[B])
        \leqslant \pi/2
        \leqslant \angmax[-\theta](\subsp[B]).
    \end{align*}Thus, the two angle intervals overlap at $\pi/2$, contradicting strict inequalities in 3). It follows that $3)\Rightarrow \srg[\theta](\subsp[A])\cap \srg[\theta](\subsp[B]\inv) = \emptyset$ and the proof is now complete.
\end{proof}
}}

\begin{figure}[t]
    \centering \vspace{-.2cm}
    \subfloat[Sectorial phase bounded \label{fig:sec-ph-side}]{%
        \parbox[c][.30\columnwidth][t]{.4\columnwidth}{%
            \centering
            \adjincludegraphics[trim={0bp 0bp 0bp 0bp},clip,height=2.8cm]{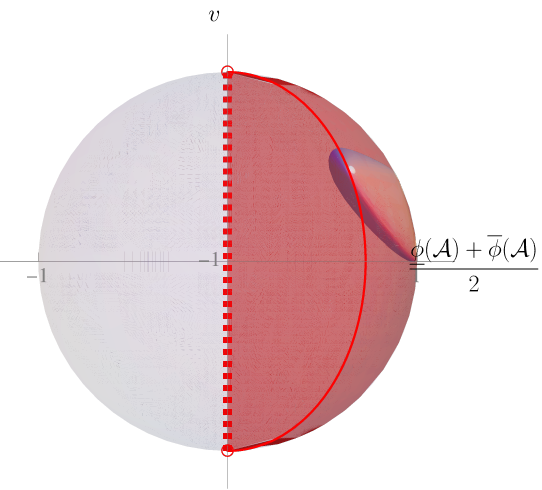}%
        }%
    }\hspace{1cm}
    \subfloat[{$\theta$-angle bounded \label{fig:seg-ph-side}}]{%
        \parbox[c][.30\columnwidth][t]{.3\columnwidth}{%
            \centering
            \adjincludegraphics[trim={0bp 0bp 0bp 0bp},clip,height=2.8cm]{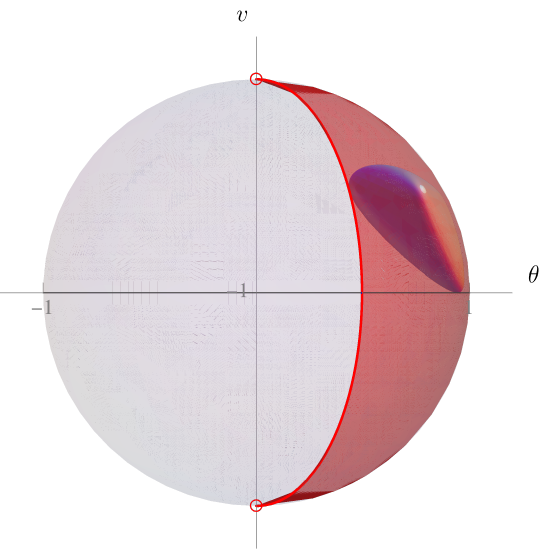}%
        }%
    }\vspace{-.4cm}
    \caption{DW Geometry and the shape of the bounding region implicitly exclude some shared boundary points.\label{fig:sep-3d-reg-side}}
    \vspace{-.6cm}
\end{figure}
{{\ignorespaces
\begin{remark} \label{rem:mat-red}
    If $\subsp[A]=\gph(A)$ and $\subsp[B]=\gph(B)$ for matrices $A,B$, all additional assumptions in \cref{prop:2d-sep,prop:1d-sep} are automatically satisfied. The gain and phase inequalities then imply the corresponding 2-D separation conditions, which in turn imply the DW-shell separation. 
    These implications, together with the necessity of $\theta$-SRG separation, establish the cross-layer relationships in \cref{fig:all-in-one-diagram}.
\end{remark}
}}

\subsubsection{{{Completing the graph of implications}}}
Note that the graph in \cref{fig:all-in-one-diagram} provides a \emph{precise} characterization of the relationships among all conditions in the {{matrix setting}}. If there is no directed path from one node to another, then the implication from the former to the latter does not hold, and counterexamples exist. Such counterexamples can often be constructed easily using normal matrices. 
For instance, the following example shows 2.2 {{and}} 3.2 {{do not imply each other.}}
\begin{example} \label{ex:sec-seg-ph}
{{\ignorespaces
    Consider the graphs of the two matrix pairs above. The pair $\gph(A_1),\gph(B_1)$ satisfies the $\theta$-angle condition for $\theta=-\pi/4$ but not the sectorial phase condition.
    Conversely, $\gph(A_2),\gph(B_2)$ satisfies the sectorial phase condition but not the $\theta$-angle condition for any $\theta$, since both DW shells intersect the polar diameter.
}}
\end{example}
\par\smallskip
    {\centering\normalfont\footnotesize
    \setlength{\tabcolsep}{5pt}
    \renewcommand{\arraystretch}{1.0}
    \begin{NiceTabular}{@{\,}l|c|c@{\,}}
        \toprule
        \Block{2-1}{\textbf{Matrices}}                               & \textbf{Sectorial phases} & \textbf{$\boldsymbol{\theta}$-angles} \\ 
            & $[\phmin(\gph(\cdot)),\phmax(\gph(\cdot))]$ & $\theta,\hspace{1em} [\angmin[\theta](\gph(\cdot)),\angmax[\theta](\gph(\cdot))]$ \\
        \midrule
        $A_1= -\ii* \oplus 1$                               & $[-\frac{\pi}{2},0]$       & $\theta = -\frac{\pi}{4},\hspace{1em} [\frac{\pi}{4},\frac{\pi}{4}]$ \\
        $B_1=\exp\left(\ii* \frac{\pi}{4} \right)\mati_2$   & $[\frac{\pi}{4},\frac{\pi}{4}]$ & $\theta = \frac{\pi}{4},\hspace{1em} [0,0]$ \\
        \midrule
        {{$A_2=\begin{bmatrix}
            1 & 2 \\
            0 & 1
        \end{bmatrix}$}} & {{$[-\frac{\pi}{2},\frac{\pi}{2}]$}} &  {{$\forall\, \theta,\quad [\frac{\pi}{2}_{-},\frac{\pi}{2}_{+}]$}} \\
        {{$B_2= 0 \oplus -1$}} & {{$[\pi, \pi]$}} & {{$\forall\, \theta,\quad [\frac{\pi}{2}_{-},\frac{\pi}{2}_{+}]$}} \\
        \bottomrule
    \end{NiceTabular}\par}
    \smallskip

{{\ignorespaces
    We also remark that the obstruction to $2.2 \Rightarrow 3.2$ arises from boundary contact with the polar diameter. If either DW shell avoids this diameter, Cond.~2.2 does imply 3.2: The sectorial phase condition yields a vertical separating plane through the diameter, with a strict gap on that shell's side. This plane can then be perturbed into an elliptic cylindrical surface (shaded dark gray in \cref{fig:dw-seg-ph}) that strictly separates the shells, establishing the $\theta$-angle Cond.~3.2 (\cref{prop:1d-sep}-3).
  
}}

\section{Graphical Stability Analysis} \label{sec:stability-ana}
{{\ignorespaces
    In this section, we derive graphical stability conditions for feedback interconnections as depicted in \cref{fig:standard-feedback}, and discuss a general formulation and pathways to solutions for feedback with uncertainties as depicted in \cref{fig:robust-feedback}.

    We first clarify what it means, in the system context, by subspaces associated with a transfer matrix and for them to admit a matrix representation. 
    \begin{defn} \label{def:sys-gph}
         Let $T(s) \in \rration^{n\times n}$ and $ N_T D_T\inv $ be its right coprime factorization over $\rhinf$. Then
        \begin{align*}
            \sgph[T](s) := \rSp(\begin{bNiceArray}{c}
                D_T(s) \\ N_T(s)
            \end{bNiceArray}) \in \grasm(n,2n) \quad\forall\, s \in \crhp\cup\{\infty\}
        \end{align*} is called the \emph{pointwise graph} of $T$ over $\crhp\cup\{\infty\}$%
        \footnote{{{These definitions are independent of the chosen ring over which the right coprime factorization is performed.}}}
        and its \emph{pointwise inverse graph} $\sinvgph[T]$ is defined by interchanging the two blocks.
    \end{defn}

    Coprimeness over $\rhinf$ guarantees $\sgph[T]$ and $\sinvgph[T]$ are well defined in $\grasm(n,2n)$ over $\crhp\cup\{\infty\}$.
    At each $s_0 \in \crhp\cup\{\infty\}$, $\sgph[T]$ is a usual subspace discussed in previous sections with basis pair $[D_T(s_0);  N_T(s_0)]$, and
    \begin{align*}
        &s_0\text{ is a zero of }T
        \Leftrightarrow \sinvgph[T](s_0)\text{ has no matrix representation}\\
        &\hspace{1.3cm}\Leftrightarrow N_T(s_0)\text{ is singular}
        \Leftrightarrow (0,-1)\in\dwshell(\sgph[T](s_0)),\\[.3ex]
        &s_0\text{ is a pole of }T
        \Leftrightarrow \sgph[T](s_0)\text{ has no matrix representation}\\
        &\hspace{1.3cm}\Leftrightarrow D_T(s_0)\text{ is singular}
        \Leftrightarrow (0,1)\in\dwshell(\sgph[T](s_0)).
    \end{align*}
    \vspace{-1.05cm}
    \subsection{Pointwise and boundary graphical stability conditions} \label{sec:stab-cond}
    An immediate consequence of \cref{lem:istab_char}-2), \cref{corol:dw-sep-fbck,prop:2d-sep} is: 
    \begin{theorem}[Pointwise stability condition] \label{thm:pwise-stab-cond}
        The positive feedback interconnection of $G, H\in \rration^{n\times n}$ is internally stable if $\dwshell(\sgph[G])\cap\dwshell(\sinvgph[H]) = \emptyset$ for all $s \in \crhp \cup\{\infty\}$, or equivalently, there exists $\theta:\crhp\cup\{\infty\}\to \rF$ such that $\srg[\theta(s)](\sgph[G]) \cap \srg[\theta(s)](\sinvgph[H]) = \emptyset$ for all $s \in \crhp\cup\{\infty\}$.
    \end{theorem}
    \begin{remark}
        The graphical conditions exclude any coincidence between a pole of $G$ and a zero of $H$ in $\crhp\cup\{\infty\}$, and vice versa. It is stronger than excluding unstable pole/zero cancellation.
        Although DW shell retains some geometric information about alignment of input and output directions of poles/zeros, it is invariant under simultaneous unitary perturbation of the input and output coordinates. The shells therefore do not determine the relative orientation of two systems' pole and zero directions. Consequently, a shared north or south pole of $\dwshell(\sgph[G])$ and $\dwshell(\sinvgph[H])$ indicates a pole/zero coincidence, but does not by itself imply cancellation.
    \end{remark}

    To relax this restriction imposed by a \emph{strict} DW separation, one can instead rely on the following weaker version:
    \begin{prop}[Weak DW separation] \label{prop:weak-sep}
        Let $[D_{\subsp[A]};N_{\subsp[A]}]$ and $[D_{\subsp[B]};N_{\subsp[B]}]$ be the respective basis pairs of $\subsp[A], \subsp[B] \in \grasm(n,2n)$. 
        Then $\subsp[A]\cap\subsp[B]\inv = \{0\}$, or equivalently, 
            $\begin{bNiceArray}{c@{\hspace{.4em}}c}
                D_{\subsp[A]} & N_{\subsp[B]} \\
                N_{\subsp[A]} & D_{\subsp[B]}
            \end{bNiceArray}$ is nonsingular,
        if $[D_{\subsp[A]}\hspace{.3em}N_{\subsp[B]}]$ and $[N_{\subsp[A]}\hspace{.3em}D_{\subsp[B]}]$ both have full row rank and  
        $\dwshell(\subsp[A])\cap\dwshell(\subsp[B]\inv) \subseteq \{(0,-1),(0,1)\}$.
    \end{prop}
    \begin{proof}
        See Appendix~\ref{pf:weak-sep}.
    \end{proof}
   The pointwise condition of \cref{thm:pwise-stab-cond} remains sufficient if the strict DW separation at poles of $G$ or $H$ in $\crhp\cup\{\infty\}$ is replaced by weak DW separation and the full row rank conditions in \cref{prop:weak-sep}.

    As foreshadowed in \cref{sec:prob_form}, we can reduce the pointwise condition on $\crhp\cup\{\infty\}$ to a boundary condition with a homotopy/continuity argument along the Nyquist contour $\nyContour$ specified above \cref{lem:istab_char}:
    \begin{theorem}[Boundary stability condition] \label{thm:bd-stab-cond}
        Let $G, H \in \rration^{n\times n}$ be proper. 
        Suppose that 
        \begin{enumerate}
            \item $\det(\mati-G(\infty)H(\infty))\neq0$,
            \item $\dwshell((\mati\oplus\mu\mati)\sgph[G](s)) \cap\dwshell(\sinvgph[H](s))=\emptyset$ for all $s\in\nyContour$ and $\mu\in[0,1]$, and
            \item $\dwshell(\sgph[G](s)) \cap\dwshell(\sinvgph[H](s))=\emptyset$ for all $s\in\poles[\ii*\rF](G)\cup\poles[\ii*\rF](H)$,
    \end{enumerate}
        then $\wno(\det(\mati-GH)) = 0$, $G\#H$ has no poles on $\ii*\,\rF\,\cup\,\{\infty\}$ and exactly $|\poles[\orhp](G)|+|\poles[\orhp](H)|$ poles in $\orhp$, counted with their McMillan multiplicities.
        In particular, if $G$ and $H$ are analytic on $\orhp$, then $G\# H$ is internally stable. 
    \end{theorem} 
    \begin{proof}
        Conds.~1) and 3) exclude closed-loop poles at $\infty$ and at imaginary-axis open-loop poles. 
        Cond.~2) (with $\mu=1$) excludes the remaining imaginary-axis closed-loop poles.

        Along the indented contour $\nyContour$, $(\mati\oplus\mu\mati)\sgph[G]$ and $\sgph[H]$ admit their respective matrix representations $\mu G, H$ for all $\mu \in [0,1]$.
            Then by \cref{thm:dw-subspace}, Cond.~2) implies $\det(\mati - \mu GH) \neq 0$ for all $s \in \nyContour, \mu\in [0,1]$.
            This guarantees the trajectory of $\det(I-\mu GH)$ along $\nyContour$ does not cross the origin when $\mu$ increases from $0$ to $1$. Since the trajectory is continuous w.r.t. $\mu$, its winding number remains unchanged: $\wno(\det(\mati-GH)) = \wno(1) = 0$. 
        Then the conclusion follows from the relationship above \cref{lem:istab_char}.
    \end{proof}
    \begin{remark}
        As in \cref{thm:pwise-stab-cond}, Conds.~2) and 3) admit equivalent formulations in terms of $\theta$-SRGs, with $\theta$ dependent on each $(s,\mu)$ in 2) and each boundary open-loop pole in 3).
        Alternatively, Conds.~2) and 3) can be replaced by the more sufficient but simpler conditions in  \cref{prop:2d-sep,prop:1d-sep}, applied pointwise along the contour and boundary open-loop poles.
    \end{remark}
    \begin{remark}
        By \cref{prop:weak-sep}, the conclusion remains valid if Cond.~3) is replaced by the following requirements at every $s\in\poles[\ii*\rF](G)\cup\poles[\ii*\rF](H)$: both $[D_G\hspace{.3em}N_H]$ and $[N_G\hspace{.3em}D_H]$ have full row rank, and $\dwshell(\sgph[G])\cap\dwshell(\sinvgph[H]) \subseteq\{(0,-1),(0,1)\}$. 
        On $\nyContour$, the rank requirements hold automatically, and weak and strict DW separation coincide because the contour excludes all poles of $G$ and $H$ and this prevents the shells from sharing a common geometric pole.
    \end{remark}

    Finally, real rationality implies $\sgph[T](\cj(s))=\cj(\sgph[T](s))$. 
    By \cref{lem:dw-shell-basics}-4), $\dwshell(\sgph[T](s))$ and $\dwshell(\sgph[T](\cj(s)))$ are therefore mirror symmetric across the $\mathrm{Re}$-$v$ plane.
    Since the verification domains in \cref{thm:pwise-stab-cond,thm:bd-stab-cond} are mirror symmetric about the real axis, it suffices to check these conditions over domains with $\tIm(s)\geqslant 0$, together with $s=\infty$ where applicable. 
    \vspace{-1em}

\subsection{On graphical robust stability analysis} \label{sec:rs_ana}
    \cref{thm:pwise-stab-cond,thm:bd-stab-cond} establish the graphical certificates for Problem~(\ref{eq:pt-gsep}).
        We hereby outline their application to \emph{graphically assisted} robust stability analysis of the interconnection depicted by \cref{fig:robust-feedback} and discuss approaches to numerical verification. 
        A detailed treatment is left to future work.

    Let $\region[X]\subseteq \ball[1]$ be closed and $\region[R] \subseteq \crhp\cup\{\infty\}$. Define the graphically characterized uncertainty class:
    \begin{align*}
        \bm{\Delta}(\region[X];\region[R]):=\{X\in \rration^{n\times n}: \dwshell(\bm{\mathcal{X}}(s))\subseteq \region[X]\,\,\forall\, s \in \region[R]\}.
    \end{align*} The bounding region $\region[X]$ may also depend on $s$.

    The analysis starts by constructing a tractable outer approximation $\bm{\Delta}(\region[X];\region[R])$ of the actual uncertainty set ${\bm{\Delta}_{\mathrm{true}}}$. 
    At each $s \in \region[R]$, models or data can guide the choice of a region $\region[X](s)$ enclosing $\bigcup_{X\in\bm{\Delta}_{\mathrm{true}}} \dwshell(\sgph[X](s))$. This formulation encompasses classical uncertainty descriptions:
    in a small-gain analysis, $\region[R]=\orhp$ and $\region[X]$ is the closed lower cap specified by a constant gain bound;
    in a positive-real analysis, $\region[R]=\orhp$ and $\region[X]$ is the closed half-ball $\{(w,v)\in\ball[1]:\tRe(w)\geqslant 0\}$. More general bounding regions can be considered to better capture the geometric features of the uncertainty set $\bm{\Delta}_{\mathrm{true}}$.

    \cref{thm:pwise-stab-cond} then gives a sufficient graphical condition for robust stability of $T$ against $\bm{\Delta}(\region[X];\orhp)$: $\dwshell(\sinvgph[T](s)) \cap \region[X] = \emptyset$ for all $s \in \crhp\cup\{\infty\}$. Although the uncertainty class is specified over $\orhp$, continuity of the subspace-based DW shells and closedness of $\region[X]$ extend their containment to $\ii*\rF\cup\{\infty\}$. 
    For certain bounding regions, such as gain caps and positive-real half-balls, containment over $\orhp$ can be characterized by boundary inequalities together with analyticity on $\orhp$. In these cases, maximum-modulus-type principles can yield boundary formulations of the robust stability conditions.

Alternatively, \cref{thm:bd-stab-cond} allows the uncertainty class to be specified through DW-shell bounds on the Nyquist contour $\nyContour$. Depending on the geometry of $\region[X]$, $\bm{\Delta}(\region[X];\nyContour)$ may be closed under scaling by $\mu\in[0,1]$. \emph{In this case}, Cond.~2) is ensured by $\dwshell(\sinvgph[T](s))\cap\region[X]=\emptyset$ for all $s\in\nyContour$. This guarantees $\wno\det(\mati-T\Delta)=0$ for every admissible $\Delta$, providing a basis for further analysis involving unstable plant and $\Delta$ as in \cite[Ch.~1]{vinnicombeUncertaintyFeedbackInfinity2001}. In addition, suppose $\bm{\Delta}_{\mathrm{true}} \subseteq\bm{\Delta}(\region[X];\ii*\rF) \cap\rhinf^{n\times n}$ {{and $T$ is analytic on $\orhp$}}.
Then $\dwshell(\sinvgph[T](\ii*\omega))\cap\region[X] =\emptyset$ for all $\omega\in[0,\infty]$ {{guarantees}} robust stability of $T$ against $\bm{\Delta}_{\mathrm{true}}$ under positive feedback.

In turn, for a given nominal system $T$, DW shells can be used to assess its robustness against different uncertainty classes. Plotting DW shells at selected frequencies provides an initial assessment of the nominal system's robustness. Numerically, the procedure in \cref{sec:dw-visual} can be applied to visualize the nominal-system shells. At each frequency, sampling $N$ directions requires solving eigenvalue problems for $N$ $n$-by-$n$ Hermitian matrices, where $n$ is the subspace dimension \cite{lestasLargeScaleHeterogeneous2012}. This gives an approximate cost of $\mathcal{O}(Nn^3)$ per frequency.
Once $\region[X]$ is chosen, $\dwshell(\sinvgph[T])\cap\region[X]=\emptyset$ is equivalent to $\dwshell(\sgph[T]) \subseteq\scomp((1\oplus-\mati_2)\region[X])$. If the latter region is regular enough (e.g., a half-ball), this containment can be expressed as a quadratic frequency-domain inequality. Over a prescribed frequency interval, the generalized KYP lemma \cite{iwasakiGeneralizedKYPLemma2005} can then convert this inequality, under its stated assumptions, into an LMI involving the state-space realization of $T$, allowing the condition to be efficiently verified. This workflow is demonstrated in \cite{liangFeedbackStabilityMixed2025}. We conclude with an example illustrating how DW-shell plots are indeed informative on robustness.
    \begin{example}
        In the setting of \cref{fig:robust-feedback}, consider $T_1$ from \cite[Ex.~2.1.1]{greenLinearRobustControl2012} and $T_2$ for comparison:
        \begin{align*}
            T_1(s) = \begin{bNiceArray}{c@{\hspace{.3em}}c}
                7 & 8 \\
                6 & 7
            \end{bNiceArray}
            \begin{bNiceArray}{c@{\hspace{.3em}}c}
                \frac{1}{s+1} & 0 \\
                0 & \frac{2}{s+2}
            \end{bNiceArray}
            \begin{bNiceArray}{c@{\hspace{.3em}}c}
                7 & 8 \\
                6 & 7
            \end{bNiceArray}\inv,\,
            T_2(s) = \begin{bNiceArray}{c@{\hspace{.3em}}c}
                \frac{1}{s+1} & \frac{s}{s+3} \\
                0 & \frac{2}{s+2}
            \end{bNiceArray}.
        \end{align*}
        The two systems have identical eigenloci given by $\frac{1}{s+1}$ and $\frac{2}{s+2}$, yielding infinite gain margin and $180^\circ$ phase margin under scalar gain and phase perturbations in negative feedback. However, their frequencywise DW shells tell a very different story. 
        We sample across positive frequencies and plot the shells of $\sinvgph[T]_1$ (\cref{fig:t1-1,fig:t1-2}) and $\sinvgph[T]_2$ (\cref{fig:t2}). 
        The sampled shells of $\sinvgph[T]_1$ rapidly fill up almost the entire ball, severely restricting the available uncertainty region, whereas those of $T_2$ leave considerably more space over the entire frequency range. This graphically indicates bad robustness of $T_1$, which agrees with the conclusion of \cite{greenLinearRobustControl2012}. 
        Meanwhile, since, in high frequency range, $\dwshell(\sinvgph[T]_2)$ contains points at every longitude and intersects every phase-bounding spherical wedge (cf.~\cref{fig:dw-sec-ph,fig:dw-seg-ph}), DW separation hence provides no robustness guarantee for $T_2$ against phase-bounded uncertainties.
    \end{example}

  \begin{figure}[htbp]
    \centering \vspace{-.3cm}
    \subfloat[Ultra low freqs.  \label{fig:t1-1}]{%
        \makebox[.35\columnwidth]{
            \adjincludegraphics[Clip={.1\width} {.2\height} {.08\width} {.15\height},height=2.42cm]{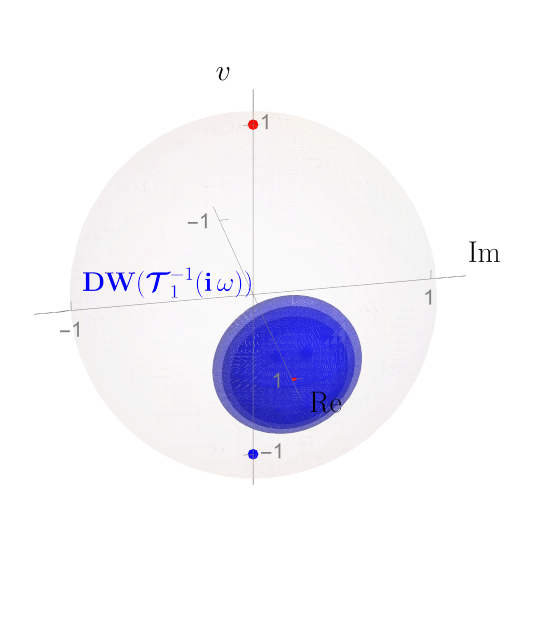}%
        }%
    }
    \subfloat[Low to mid freqs. \label{fig:t1-2}]{%
        \makebox[.34\columnwidth]{
            \adjincludegraphics[Clip={.1\width} {.2\height} {.08\width} {.15\height},height=2.42cm]{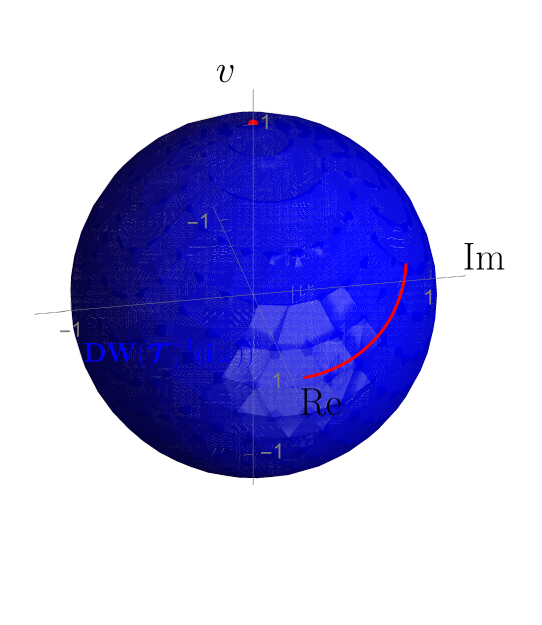}%
        }%
    }
    \subfloat[Full range \label{fig:t2}]{%
        \makebox[.32\columnwidth]{
            \adjincludegraphics[Clip={.1\width} {.2\height} {.08\width} {.15\height},height=2.42cm]{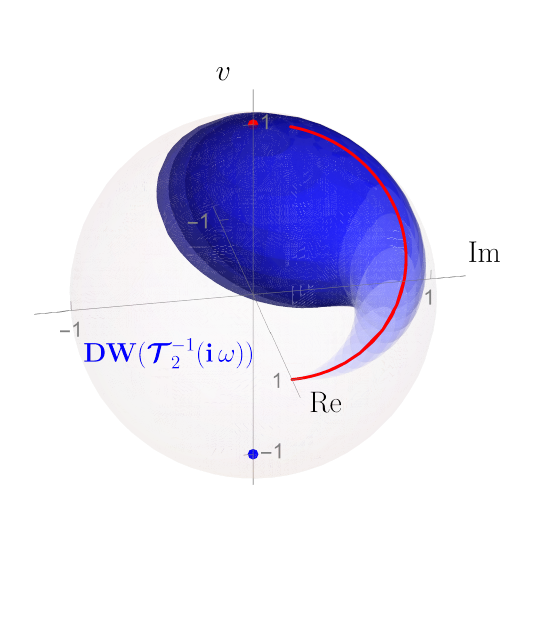}%
        }%
    }\vspace{-.2cm}
    \caption{(a)-(b): $\dwshell(\sinvgph[T]_1)$;\hspace{1em} (c):$\dwshell(\sinvgph[T]_2)$. \label{fig:rs-ex}}
    \vspace{-.4cm}
\end{figure}

}}

\vspace{-1em}
\section{Conclusions} \label{sec:concl}
We present a unified {{perspective}} for the graphical stability analysis of MIMO LTI feedback systems, grounded in the concept of DW shells{{ of linear subspaces}}. 
By leveraging the properties of DW shells and various projection maps, we {{establish}} an intuitive geometric understanding of various existing graphical tools and
gain/phase notions. This perspective provides a cohesive interpretation of existing {{graphical}} stability criteria for bi-component feedback loops and reveals their relative conservatism via a shared geometric foundation.
Another contribution is the introduction of $\theta$-SRG, which captures both gain and (segmental) phase information in a unified representation. {{Its separation condition is equivalent to DW-shell separation and is therefore the least restrictive among the SRG-type conditions.}}
{{We derive pointwise and boundary stability conditions suitable for graphically assisted robust stability analysis and a weak DW separation result that relaxes strict separation at open-loop unstable poles.}}

\appendices
\makeatletter
\patchcmd{\@IEEEprocessthesectionargument}
  {\thesectiondis \\*}
  {\thesectiondis:\ }
  {}
  {\PackageWarning{appendix-layout}{Heading patch failed}}
\makeatother
{{\ignorespaces
\section{Proof of Lemma~\ref{lem:unary-op}} \label{app:proof-dw-basic}
We use the real form~(\ref{defn-dw-r3}) of the DW shell.
Items 1), 3) and 4) are straightforward from the definition. We provide proofs for the less obvious 2) and 5).

{$\mathsf{2)}$:} Referring to (\ref{defn-dw-r3}), let
    $\tilde{E}_1 := \begin{bNiceArray}{c@{\hspace{.5em}}c}
                0 & 1 \\[-2pt]
                1 & 0
            \end{bNiceArray}, 
           \tilde{E}_2 := \begin{bNiceArray}{c@{\hspace{.5em}}c}
                0 & -\ii* \\[-2pt]
                \ii* & 0
            \end{bNiceArray},
            \tilde{E}_3 := \begin{bNiceArray}{c@{\hspace{.5em}}c}
                -1 & 0 \\[-2pt]
                0 & 1
            \end{bNiceArray}.$
Note that $\tilde{E}_1, \tilde{E}_2, \tilde{E}_3$ constitute an orthogonal basis of the real vector space of 2-by-2 Hermitian matrices with zero trace.
Substituting the transformed subspace $(U \otimes \mati_n)\subsp[A]$ to (\ref{defn-dw-r3}) yields the following form of $\dwshell((U \otimes \mati_n)\subsp[A])$:
\begin{align*}
    \{(\xi\ct\,((U\ct \tilde{E}_i U) \otimes \mati_n) \xi; i=1,2,3): \xi \in \subsp[A], \xi\ct \xi = 1\},
\end{align*}
and observe that, for $i=1,2,3$, $U\ct \tilde{E}_i U$ is within the abovesaid space and can be uniquely written as a linear combination of $\tilde{E}_1, \tilde{E}_2, \tilde{E}_3$ with real coefficients. This guarantees there is a unique $Q(U) \in \rF^{3\times 3}$ such that 
$\dwshell((U \otimes \mati_n)\subsp[A]) =  Q(U)\, \dwshell(\subsp[A])$.
Moreover, $Q(U)$ has an explicit formula $[Q(U)]_{ij} = \frac{1}{2}\Tr(U\ct \tilde{E}_i U \tilde{E}_j)$ and is indeed orthonormal.

{$\mathsf{5)}$:} The dimension of the ambient space is $2n$, so $\subsp[A]\in\grasm(n,2n)$ implies the dimension of its complement is $n$, i.e., $\subsp[A]^\perp \in \grasm(n,2n)$.
Let $[D;N]$ be a basis pair of $\subsp[A]$. If $D$ is invertible, then $\subsp[A] = \gph(A)$ with $A = ND\inv$.
    Note that 
    \begin{align*}
        \begin{bNiceArray}{c}
            \mati \\ A
        \end{bNiceArray}\ct
        \begin{bNiceArray}{c}
            -A\ct \\ \mati
        \end{bNiceArray} = \mato \Rightarrow
        \subsp[A]^\perp = (\begin{bNiceMatrix}
            0 & -1 \\ 1 & 0
        \end{bNiceMatrix}\otimes I)\gph(A\ct).
    \end{align*}
    Moreover, $\left[\begin{smallmatrix}
            0 & -1 \\ 1 & 0
        \end{smallmatrix}\right] = -\left[\begin{smallmatrix}
            1 & 0 \\ 0 & e^{\ii*\pi}
        \end{smallmatrix}\right]
        \left[\begin{smallmatrix}
            0 & 1 \\ 1 & 0
        \end{smallmatrix}\right]$, and it follows from Properties 1) and 2) that 
        $\dwshell(\subsp[A]^\perp) = (-1\oplus 1 \oplus -1)\dwshell(A\ct)$.
    By \cite[Thm.~3.7]{liDavisWielandtShellsOperators2008}, $\dwshell(A\ct) = (1\oplus -1 \oplus 1)\dwshell(A)$.
    Thus, we have $\dwshell(\subsp[A]^\perp) = - \dwshell(A) = - \dwshell(\subsp[A])$.

If $D$ is not invertible, we can always construct a sequence of invertible $D+\vare_k \mati$ with $\vare_k \to 0$ as $k\to\infty$. This yields a sequence of perturbed subspaces $\subsp[A]_k: = \rSp([(D+\vare_k \mati)\ct \hspace{.5em} N\ct]\ct)$ that converges to $\subsp[A]$, while $\subsp[A]_k^\perp \to \subsp[A]^\perp$.
By Definition~(\ref{defn-dw-r3}), the DW shells depend continuously on subspaces in the Hausdorff metric, and thus
$
    \dwshell(\subsp[A]^\perp) = \lim_{k\to\infty} \dwshell(\subsp[A]_k^\perp) 
    = \lim_{k\to\infty} (- \dwshell(\subsp[A]_k)) = - \dwshell(\subsp[A]).
$ 

}}

\vspace{-1em}
\section{Proof of Theorem \ref{thm:dw-subspace}} \label{app:proof-dw-sep-ext} 
{{\ignorespaces
We start by proving a technical lemma.
\begin{lemma} \label{prop:sim-align}
    Let $u_1, u_2, v_1, v_2 \in \cF^n$ be unit vectors. If $\iprod(u_1)(u_2) = \iprod(v_1)(v_2)$, then there exists a unitary $W$ such that $W [u_1\hspace{.5em}u_2] = [v_1\hspace{.5em}v_2]$.
\end{lemma}
\begin{proof}
    For $n=1$, we simply let $W = v_1\cj(u_1)$. For $n\geqslant 2$, the inner product equality guarantees that $[u_1\hspace{.5em}u_2]$ and $[v_1\hspace{.5em}v_2]$ admit QR decompositions with a common R factor: $[u_1\hspace{.5em}u_2] = Q_u R, [v_1\hspace{.5em}v_2]=Q_v R$ where
    \begin{align*}
        R = \begin{bNiceMatrix}
            1 & \iprod(u_1)(u_2) \\
            0 & \sqrt{1-|\iprod(u_1)(u_2)|^2}
        \end{bNiceMatrix} = \begin{bNiceMatrix}
            1 & \iprod(v_1)(v_2) \\
            0 & \sqrt{1-|\iprod(v_1)(v_2)|^2}
        \end{bNiceMatrix}, 
    \end{align*} and $Q_u, Q_v \in\cF^{n\times 2}$ are isometries, which can be extended to unitary matrices: $W_u := [Q_u\hspace{.5em} Q_u^\perp], W_v := [Q_v\hspace{.5em}Q_v^\perp]$. Let $W:= W_v W_u\ct$, then we have $W[u_1 \hspace{.3em}u_2] = W_v (W_u\ct Q_u R) = Q_v R = [v_1\hspace{.3em}v_2]$, which completes the proof. 
\end{proof}

We continue with the proof of \cref{thm:dw-subspace}. By the dimensionality assumption, matrices $X, Y$ must have full column rank; the equivalence between 2) and 3) is straightforward.

For ``1)\,$\Rightarrow$2)\,'': Suppose that 2) does not hold so there exists a unitary $U$ and a nonzero $x = [x_1\ct \hspace{.5em}x_2\ct]\ct \in \mathcal{S}$ such that 
\begin{align*}
    \begin{bNiceMatrix}
        y_1 \\ y_2
    \end{bNiceMatrix}:= (\mati_2 \otimes U)x = 
    \begin{bNiceMatrix}
        Ux_1 \\ Ux_2
    \end{bNiceMatrix} \in \mathcal{T}\setminus \{0\}.
\end{align*}
Since $U$ is unitary, $\iprod(x_1)(x_2) = \iprod(U x_1)( U x_2)$, $\|x_i\| = \|U x_i \|$ for $i = 1,2$, and we have:
\begin{align}
     \frac{2\iprod(x_1)(x_2)}{\|x_1\|^2+\|x_2\|^2} &= 
         \frac{2\iprod(y_1)(y_2)}{\|y_1\|^2+\|y_2\|^2}, \label{eq:dw-first}\\
     \frac{\|x_2\|^2-\|x_1\|^2}{\|x_2\|^2+\|x_1\|^2} &=
         \frac{\|y_2\|^2-\|y_1\|^2}{\|y_2\|^2+\|y_1\|^2}. \label{eq:dw-third} 
\end{align}
It follows that $\left(\frac{2\iprod(x_1)(x_2)}{\|x_1\|^2+\|x_2\|^2}, \frac{\|x_2\|^2-\|x_1\|^2}{\|x_2\|^2+\|x_1\|^2}\right)$ lies in $\dwshell(\mathcal{S}) \cap \dwshell(\mathcal{T})$, which contradicts 1).

For ``2)\,$\Rightarrow$\,1)'': Suppose that 1) does not hold so there exist $x = [x_1\ct\hspace{.5em}x_2\ct]\ct \in \mathcal{S}\setminus\{0\}, y = [y_1\ct\hspace{.5em}y_2\ct]\ct\in \mathcal{T}\setminus \{0\}$ such that (\ref{eq:dw-first}) and (\ref{eq:dw-third}) hold simultaneously.
We next show that in such cases it is always possible to find a unitary $U$ so that $(\mati_2 \otimes U)x$ aligns with $y$, i.e., $Ux_1= \mu y_1, Ux_2 = \mu y_2$ for some $\mu$.
Note that (\ref{eq:dw-third}) is equivalent to 
\begin{align*}
    \textstyle \|x_1\| \|y_2\| = \|x_2\| \|y_1\|. \numberthis \label{eq:norm-eq}
\end{align*}
For two nonzero vectors, it is easy to find a unitary matrix to align them.
If $x_1 = 0$, then $x_2 \neq 0$ since $x\neq 0$, and $y_1$ must be zero from the above norm equality, which in turn implies $y_2\neq 0$ as $y\neq 0$. In this case, the simultaneous alignment problem reduces to alignment of two nonzero vectors $x_2, y_2$. 
A similar argument works for the case with $x_2 =0$.

The remaining case is when $x_i, y_i \neq 0\, (i = 1,2)$.
In this case, (\ref{eq:norm-eq}) can be rewritten as $\frac{\|x_1\|}{\|x_2\|} = \frac{\|y_1\|}{\|y_2\|}>0$, in light of which (\ref{eq:dw-first}) simplifies to
$
    \langle\frac{x_1}{\|x_1\|},\frac{x_2}{\|x_2\|}\rangle = 
            \langle\frac{y_1}{\|y_1\|},\frac{y_2}{\|y_2\|}\rangle.
$
From \cref{prop:sim-align}, there always exists a unitary $U$ that makes $Ux_1 = \frac{\|x_1\|}{\|y_1\|} y_1, Ux_2 = \frac{\|x_2\|}{\|y_2\|}y_2$, where $\frac{\|x_1\|}{\|y_1\|} = \frac{\|x_2\|}{\|y_2\|}$ is again due to the norm equality. 
The proof is now complete.  
\section{Proof of Proposition~\ref{prop:weak-sep}} \label{pf:weak-sep}
Suppose that  $\subsp[A] \cap \subsp[B]\inv$ contains a nonzero $[u\ct\hspace{.3em}y\ct]\ct$. We show that the weak DW separation forces one of the block rows to lose rank.
Since the two shells can only intersect at the two poles, either $u=0$ or $y=0$. Assume without loss of generality (WLOG) that $u=0$ and $\|y\|=1$ and also that $[D_{\subsp[A]};N_{\subsp[A]}], [D_{\subsp[B]};N_{\subsp[B]}]$ are normalized as above \cref{dw:triple-herm}. 
Then the shells both contain the north pole and there exist $x_{\subsp[A]}, x_{\subsp[B]} \neq 0$ such that $y = N_{\subsp[A]} x_{\subsp[A]} = D_{\subsp[B]} x_{\subsp[B]}$ and $0=D_{\subsp[A]}x_{\subsp[A]} = N_{\subsp[B]} x_{\subsp[B]}$.

\emph{For $n=1$ or $n\geqslant 3$}, by convexity, the weak DW separation ensures a separating hyperplane passing through $(0,1)$: 
$\{(w,v): \tRe(\cj(\eta) w)+\nu (v-1) = 0\}$ with a nonzero normal $(\eta,\nu)$.
Since the hyperplane is a supporting plane of both shells at $(0,1)$, by \cref{dw:triple-herm} the following matrix
\begin{align*}
    &X\ct\, (\tRe(\eta)E_1 + \tIm(\eta)E_2 + \nu E_3) X - \nu \mati \\
    &=
    X\ct \begin{bNiceArray}{c@{\hspace{.3em}}c}
        -2\nu\mati & \cj(\eta)\mati \\
        \eta\mati & \mato
    \end{bNiceArray} X\text{ is semidefinite and singular, and }\\
    &\hspace{1em}
    \begin{bNiceMatrix}
        0 \\ y
    \end{bNiceMatrix} \in \nSp(X\ct \begin{bNiceArray}{c@{\hspace{.3em}}c}
        -2\nu\mati & \cj(\eta)\mati \\
        \eta\mati & \mato
    \end{bNiceArray} )\text{, for each }
    X \in 
    \{\begin{bNiceMatrix}
        D_{\subsp[A]}\\ N_{\subsp[A]}
    \end{bNiceMatrix},
    \begin{bNiceMatrix}
        N_{\subsp[B]}\\ D_{\subsp[B]}
    \end{bNiceMatrix}
    \}.
\end{align*}  
If $\eta \neq 0$, then the above kernel condition implies $D_{\subsp[A]}\ct y = 0$ and $N_{\subsp[B]}\ct y = 0$, i.e., $y\ct [D_{\subsp[A]}\hspace{.3em}N_{\subsp[B]}] =0$, contradicting the full row rank condition. 
If $\eta = 0$, then the plane is horizontal through the north pole, forcing the shell of one subspace to be a singleton at the north pole. Assuming that subspace is $\subsp[A]$, then $D_{\subsp[A]} = 0$, while $N_{\subsp[B]}$ is singular since $(0,1)\in \dwshell(\subsp[B]\inv)$. Thus, $[D_{\subsp[A]}\hspace{.3em}N_{\subsp[B]}]$ loses row rank, which is again a contradiction.

\emph{For $n=2$}, a DW shell may be an ellipsoid surface. The preceding argument applies if the two convex hulls admit a weak separating plane. It remains to exclude the nesting of two shells, with contact confined to the poles.
We show that this cannot happen under the full row rank condition.
Assume WLOG that $\convhull(\dwshell(\subsp[A])) \subseteq \convhull(\dwshell(\subsp[B]\inv))$. 
Construct unitary $U = [y\hspace{.3em}y_\perp]$ and nonsingular upper triangular matrices $S_{\subsp[A]},S_{\subsp[B]}$. Then a common unitary change of input-output coordinates and independent changes of basis give the following form, preserving the DW shells and all relevant ranks:
\begin{align*}
    (\mati_2\otimes U\ct)\begin{bNiceArray}{c@{\hspace{.3em}}c}
        D_{\subsp[A]} & N_{\subsp[B]}\\
        N_{\subsp[A]} & D_{\subsp[B]}
    \end{bNiceArray}(S_{\subsp[A]} \oplus S_{\subsp[B]})
    =
    \begin{bNiceArray}{c@{\hspace{.5em}}c:c@{\hspace{.5em}}c}
        0 & a_1 & 0 & c_1\\
        0 & a_2 & 0 & c_2\\ \hdottedline
        1 & 0   & 1 & 0\\
        0 & \alpha & 0  & \gamma
    \end{bNiceArray},
\end{align*} 
If $\alpha = 0$, then $\dwshell(\subsp[A])$ contains the south pole. This would force $\dwshell(\subsp[B]\inv)$ to also contain the south pole and thus $\gamma = 0$, which causes the second block row to violate the full row rank condition. 
Thus $\alpha \neq 0$, and necessarily $a_2=0$, for otherwise $\subsp[A]$ has a finite nonzero spectral point which yields a nonpolar DW point on the unit sphere and precludes the nesting.
Now for the transformed block matrix, since each block column and block row have full rank, 
$\left[\begin{smallmatrix}
    a_1 \\ a_2
\end{smallmatrix}\right]$ and $\left[\begin{smallmatrix}
    c_1 \\ c_2
\end{smallmatrix}\right]$ are nonzero and linearly independent. Therefore, $a_2=0 \Rightarrow a_1 \neq 0 \Rightarrow c_2 \neq 0$. 
Then the remaining case is when $\alpha, a_1, c_2\neq 0$, $a_2 = 0$. 
In this case, with \cref{dw:triple-herm} and some algebraic manipulation, the horizontal cross sections of $\dwshell(\subsp[A])$ and $\dwshell(\subsp[B]\inv)$ at $v = 1-\delta \in [-1,1]$, if non-empty, are circles with radii:
\begin{align*}
    \textstyle r_{\subsp[A]}(\delta) = \sqrt{2\delta - k_{\subsp[A]}\delta^2 },\hspace{1em}
    r_{\subsp[B]\inv}(\delta) =
    \sqrt{
        \frac{2|c_1|^2}{|c_1|^2+|c_2|^2}\delta  - k_{\subsp[B]}\delta^2 
    }
\end{align*} where $k_{\subsp[A]}, k_{\subsp[B]}$ are positive constants.
Note that 
    $\lim_{\delta\to0} \frac{r_{\subsp[A]}(\delta)}{\sqrt{2\delta}} = 1 > 
    \lim_{\delta\to 0} \frac{r_{\subsp[B]\inv}(\delta)}{\sqrt{2\delta}}$,
which indicates that $\dwshell(\subsp[A])$ cannot be nested within $\dwshell(\subsp[B]\inv)$. The proof is now complete.

}}

\section*{References}
\def\refname{\vadjust{\vspace*{-2.5em}}} %
\vspace{-.5cm}
\bibliographystyle{IEEEtran}

\vspace{-.7cm}
\begin{IEEEbiography}{Ding Zhang} 
    received the B.Eng. degree in Mechanical Engineering from Huazhong University of Science and Technology, Wuhan, China, in 2018, and the Ph.D. degree in Electronics and Computer Engineering from the Hong Kong University of Science and Technology, Hong Kong SAR, in 2024, where he {{worked as}} a postdoctoral research fellow. He was a visiting researcher with the King Abdullah University of Science and Technology (KSA), the University of Cambridge (UK){{, and Nanjing University (CN)}}.
    {{He is currently a research fellow with the Australian National University, Canberra, Australia.}}
    His research focuses on graphical stability analysis of large-scale networked and uncertain systems, with broader interests in matrix theory, spectral graph theory, and control theory.
     \vspace{-4em}
\end{IEEEbiography}

\begin{IEEEbiography}{Xiaokan Yang} 
    received the B.S. degree in Vehicle Engineering from Tsinghua University, Beijing, China, in 2021{{. H}}e is currently pursuing the Ph.D. degree in dynamical system and control with the Department of Mechanics and Engineering Science at Peking University. His research interests include linear systems and control, phase theory, LMIs and dissipativity theory. 
    \vspace{-4em}
\end{IEEEbiography}

\begin{IEEEbiography}{Axel Ringh} received the M.Sc. degree in Engineering Physics in 2014, and the Ph.D. degree in Applied and Computational Mathematics in 2019, both from KTH Royal Institute of Technology, Stockholm,
Sweden. From 2019 to 2021 he was a postdoctoral researcher with the Department of Electronic and Computer Engineering, the Hong Kong University of Science and Technology, Hong Kong S.A.R, China, and since 2021 he {{has been}} with the Department of Mathematical Sciences, Chalmers University of Technology and University of Gothenburg, Gothenburg, Sweden (2021-2025 as assistant professor, and from 2025 as associate professor). {{He}} is the recipient of the European Control Conference 2015 Best Student Paper Award, and one of the recipients of the SIAM Activity Group on Control and Systems Theory Best SICON Paper Prize 2023. His research interests are within field of applied mathematics, specifically within the areas of optimization and systems theory, with applications to problems in control theory, signal processing, inverse problems, and machine learning.
\vspace{-4em}
\end{IEEEbiography}

\begin{IEEEbiography}{Li Qiu} (Fellow, IEEE) received his Ph.D. degree in electrical engineering from the University of Toronto in 1990. He was with the Canadian Space Agency, the Fields Institute for Research in Mathematical Sciences (Waterloo), and the Institute of Mathematics and its Applications (Minneapolis). In 1993, he joined the Hong Kong University of Science and Technology, Hong Kong, where he worked through the ranks until 2024. He is currently a Presidential Chair Professor with the Chinese University of Hong Kong, Shenzhen, China. His research interests include system, control, optimization theory and mathematics for information technology as well as their applications in manufacturing industry and energy systems. He served as an associate editor of the IEEE Transactions on Automatic Control and an associate editor of Automatica. He was the general chair of the 2009 7th Asian Control Conference. He was a Distinguished Lecturer from 2007 to 2010 and was a member of the Board of Governors in 2012 and 2017 of the IEEE Control Systems Society. He served as a member of the steering committee and as a vice president of Asian Control Association. He is a member of the steering committees of the International Symposiums of Mathematical Theory of Networks and Systems. He is the founding chairperson of the Hong Kong Automatic Control Association. He is a Fellow of IEEE and a Fellow of IFAC.
\end{IEEEbiography}

\end{document}